\documentclass[11pt]{memoir}
\usepackage[cp1251]{inputenc}
\usepackage[english]{babel}
\usepackage{amssymb,amsmath,amscd,amsthm,latexsym}
\usepackage[matrix,arrow,curve]{xy}
\CompileMatrices
\usepackage{tikz}
\usetikzlibrary{matrix,arrows}
\usepackage{tikz-cd}
\usepackage{prooftree}
\usepackage{array}
\def\ruleone#1#2{\prooftree #1 \justifies #2 \endprooftree}
\def\ruletwo#1#2#3{\ruleone{#1\quad #2}{#3}}
\def\ruledot#1{\prooftree \proofdotseparation=1.2ex\proofdotnumber=3
               \leadsto #1 \justifies #1 \endprooftree}
\theoremstyle{plain}
\newtheorem{theorem}{Theorem}[section]
\newtheorem{fact}[theorem]{Fact}
\newtheorem{corollary}[theorem]{Corollary}
\newtheorem{lemma}[theorem]{Lemma}
\newtheorem{proposition}[theorem]{Proposition}
\theoremstyle{definition}
\newtheorem{definition}[theorem]{Definition}
\newtheorem{convention}[theorem]{Convention}

\newtheorem{note}[theorem]{Note}
\newtheorem{example}[theorem]{Example}

\numberwithin{equation}{section}
\newcommand{\tri}{\triangleright}
\newcommand{\x}{\mathsf{x}}
\newcommand{\y}{\mathsf{y}}
\newcommand{\z}{\mathsf{z}}

\newcommand{\Up}{\Uparrow\!}

\newcommand{\D}{\mathsf{D}}
\newcommand{\un}{\underline{0}}
\newcommand{\ulam}{\boldsymbol{\lambda}}

\begin{document}
\title{Is $\alpha$-conversion easy?}
\author{George Cherevichenko}
\date{}
\maketitle
We present a new $\lambda$-calculus  with explicit
substitutions and named variables.  Renaming of bound variables in this calculus is explicit (there is a special rewrite rule) and can be delayed. Contexts (environments) are not sets or lists without multiplicity, but have a more complicated structure. There is a natural order on the set of contexts. A ``set" of free variables is not a set, but a context in the new sense. New definitions simplify working with $\alpha$-conversion.

\newpage

\section{Introduction, $\alpha$-conversion for the usual $\lambda$-terms}

\begin{figure}
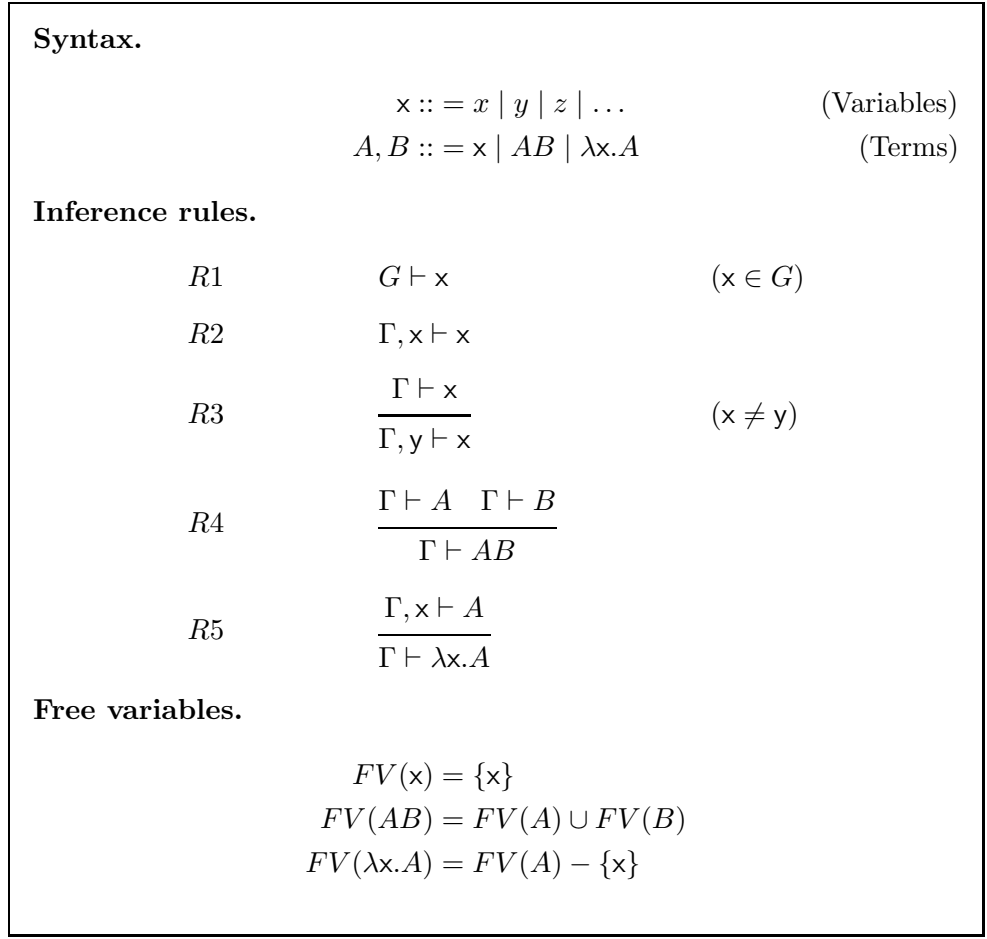

\begin{framed}
\noindent
\textbf{Syntax.}
\begin{align*}
\x::&= x\mid y\mid z\mid\ldots \tag{Variables}\\
A,B::&= \x \mid  AB \mid \lambda\x.A   \tag{Terms}
\end{align*}

\textbf{Inference rules.}
\begin{align*}
&R1 && G\vdash\x && (\x\in G)\\[5pt]
&R2&& \Gamma,\x\vdash \x &\\[5pt]
&R3&& \ruleone{\Gamma\vdash
\x}{\Gamma,\y\vdash \x} && (\x\neq \y) \\[10pt]
&R4&& \ruletwo{\Gamma\vdash A}{\Gamma\vdash B}{\Gamma\vdash
AB} &\\[10pt]
&R5&& \ruleone{\Gamma,\x\vdash A}{\Gamma\vdash\lambda \x.A}
\end{align*}

\textbf{Free variables.}
\begin{align*}
FV(\x)&=\{\x\}\\
FV(AB)&=FV(A)\cup FV(B)\\
FV(\lambda\x.A)&=FV(A)-\{\x\}
\end{align*}

\end{framed}
\caption{Terms and inference rules}\label{inff1}
\end{figure}

\begin{convention}The symbols $x,y,z\ldots$ are \emph{variables}. The symbols $\x,\y,\z$ range over variables. The inequality $\x\neq\y$ means that $\x$ and $\y$ denote different variables.
\end{convention}
\begin{definition} A \emph{global context} is a possibly empty, finite set of variables.
\end{definition}
\begin{definition}
A \emph{local context} is a possibly empty, finite list of variables
with multiplicity (i.e., repetitions are permitted).
\end{definition}
\begin{example}
The list $x,x,y$ is a local context.
\end{example}
\begin{definition}A \emph{context} is a pair $G,L$, where $G$ is a global context  and $L$ is a local context. The symbols
$\Gamma,\Delta,\Sigma$ range over contexts. If $\Gamma$ is $G,L$ then $\x\in\Gamma$ is shorthand for $\x\in G\vee\x\in L$.
\end{definition}
\begin{example} $\{x,z\},x,x,y$ is a context, where $\{x,z\}$ is a global context  and $x,x,y$ is a local context.
\end{example}
\begin{convention}If $\,\Gamma$ is  $G,L$ then   $\Gamma,\x$ is shorthand for $G,L,\x$ (the global context is $G$ and the local context is $L,\x$).
\end{convention}

\begin{definition}\emph{Terms}  are defined  on Figure~\ref{inff1}.\\
A \emph{judgement} is an expression of the form $\Gamma\vdash A$.\\ Inference rules for judgements are shown on Figure~\ref{inff1} where $G\vdash \x$ is shorthand for $G,nil\vdash \x$  and $nil$ is the empty list.
\end{definition}

\begin{example}
$$
\ruleone{\ruleone{\emptyset,x,x\vdash x}{\emptyset,x\vdash\lambda x.x}}{\emptyset\vdash\lambda
x.\lambda x.x}
$$
\end{example}
\begin{example}
$$\ruleone{\ruleone{\ruletwo{\ruleone{\emptyset,x\vdash x}{\emptyset,x,y\vdash x}}{\emptyset,x,y\vdash y}{\emptyset,x,y\vdash xy}}{\emptyset,x\vdash\lambda y.xy}}{\emptyset\vdash\lambda x.\lambda y.xy}$$
\end{example}
\begin{example}
$$\ruleone{\ruletwo{\ruleone{\{x\}\vdash x}{\{x\},y\vdash x}}{\{x\},y\vdash y}{\{x\},y\vdash xy}}{\{x\}\vdash\lambda y.xy}$$
\end{example}

In the similar calculus with types the judgement $\Gamma\vdash \x:T$ means ``the rightmost occurrence of $\x$ in $\Gamma$ has  type $T\,$". Hence, we can derive\\
$\{x: \mathbb{N}, y: \mathbb{N}\},y: \mathbb{R}\vdash x: \mathbb{N}$

$$\ruleone{\{x: \mathbb{N}, y: \mathbb{N}\}\vdash x: \mathbb{N}} {\{x: \mathbb{N}, y: \mathbb{N}\}, y: \mathbb{R} \vdash x: \mathbb{N}}$$
and $\{x: \mathbb{N}, y: \mathbb{N}\},y: \mathbb{R}\vdash y: \mathbb{R}$\\
but we can not derive\\
$\{x: \mathbb{N}, y: \mathbb{N}\},y: \mathbb{R}\vdash y: \mathbb{N}$

 \begin{lemma}[Generation lemma]$ $\\
Each derivation of $\,G\vdash \x$ is an application of the
rule~$R1$.\\
Each derivation of $\,\Gamma,\x\vdash \x$ is an application of the
rule~$R2$.\\
Each derivation of $\,\Gamma,\y\vdash \x$ $($where $\x\neq \y)$ is
an application of the rule~$R3$ to some derivation of
$\,\Gamma\vdash
\x$.\\
Each derivation of $\,\Gamma\vdash AB$ is an application of the
rule~$R4$ to some  derivations of $\,\Gamma\vdash A$ and
$\,\Gamma\vdash B$.\\
Each derivation of $\,\Gamma\vdash \lambda \x.A$  is an
application of the rule~$R5$ to some derivation of
$\,\Gamma,\x\vdash A$.\\
\end{lemma}
\begin{proof}
The proof is straightforward.\\
\end{proof}

\begin{lemma}\label{lemmaax}$\Gamma\vdash\x$ is derivable iff $\x\in\Gamma$.
 \end{lemma}
 \begin{proof}Induction over the length of local part of $\Gamma$. If this length is equal to $0$, then $\Gamma$ has the form $G$ and $G\vdash\x$ is derivable iff $\x\in G$. If $\Gamma$ has the form $\Delta,\y$ then either $\x=\y$  or $\x\neq\y$. In the first case $\Gamma\vdash\x$ is derivable and $\x\in\Gamma$. In the last case $\Delta,\y\vdash\x$ is derivable iff $\Delta\vdash\x$ is derivable and we use the induction hypothesis.\\
 \end{proof}

\begin{proposition} For any derivable judgement, its
derivation is unique.
\end{proposition}
\begin{proof} We  construct the derivation from the bottom up, using Generation lemma.
\end{proof}

We obtain this result because there are no weakening rules except of $R3$.

\begin{fact} $FV(A)\vdash A$ is derivable for each term $A$.
\end{fact}

\begin{figure}
\begin{framed}
\noindent

\textbf{Syntax.}
\begin{align*}
\x::&= x\mid y\mid z\mid\ldots \tag{Variables}\\
a,b::&= \x \mid \un \mid a[\uparrow] \mid  ab \mid \lambda a   \tag{``De Brujn's terms''}
\end{align*}

\textbf{Correspondence.}
\begin{align*}
& &&\|G\vdash \x\|=  \x & (\x\in G)\\[15pt]
&&&\|\Gamma,\x\vdash \x\|=  \un \\[15pt]
&&&\ruleone{\|\Gamma\vdash
\x\|=  a}{\|\Gamma,\y\vdash \x\|=   a[\uparrow]} & (\x\neq \y)\\[20pt]
&&&\ruletwo{\|\Gamma\vdash A\|=   a}{\|\Gamma\vdash B\|=  b}{\|\Gamma\vdash
AB\|=   ab}\\[15pt]
&&&\ruleone{\|\Gamma,\x\vdash A\|=   a}{\|\Gamma\vdash\lambda \x.A\|=    \lambda a}
\end{align*}
\end{framed}
\caption{Correspondence}\label{correspondencee}
\end{figure}

\begin{definition}We associate with every derivable judgement $\Gamma\vdash A$  some
``De Brujn's term'' $\|\Gamma\vdash A\|$ as it is shown on Figure~\ref{correspondencee} (by recursion over the unique
derivation of $\Gamma\vdash A$).
\end{definition}

\begin{example}
$$
\ruleone{\|\emptyset,x\vdash x\|=\un}{\|\emptyset\vdash\lambda x.x\|=\lambda\un}
$$
\end{example}

\begin{example}
$$\ruleone{\ruletwo{\ruleone{\|\{x\}\vdash x\|= x}{\|\{x\},y\vdash x\|= x[\uparrow]}}{\|\{x\},y\vdash
y\|=\un}{\|\{x\},y\vdash xy\|= (x[\uparrow])\un}}{\|\{x\}\vdash\lambda y.xy\|=
\lambda(x[\uparrow])\un}$$
\end{example}

   \begin{definition}We write $A\equiv_{\alpha}B$ iff
 $\|FV(A)\vdash A\|$ is the same ``De Brujn's term'' as $\|FV(B)\vdash B\|$.
  \end{definition}

  \begin{example} $\lambda y.xy\equiv_{\alpha}\lambda z.xz$,  because\\
  $\|\{x\}\vdash\lambda y.xy\|=
\lambda(x[\uparrow])\un=
\|\{x\}\vdash\lambda z.xz\|$
  \end{example}

 Now we define a partial order on the set of contexts. We want $FV(A)$ to be ``the smallest" context $\Gamma$ such that $\Gamma\vdash A$ is derivable. We want $FV(x)=\{x\}$. But $\emptyset,x\vdash x$ is derivable too, hence we want $\{x\}<\emptyset,x$

 \begin{definition}The order $\leqslant$  on the set of all contexts is the smallest partial order with the following properties:
 \begin{enumerate}
   \item $G,L< G\cup\{\x\},L\qquad(\x\not\in G)$
   \item $G,L< (G-\{\x\}),\x,L$
    \end{enumerate}
 \end{definition}

 \begin{example}$\{z\},y<\{z,x\},y<\{z\},x,y<\{z,x\},x,y$
 \end{example}

 \begin{proposition}$G_1,L_1\leqslant G_2,L_2$ iff $L_2=LL_1$ for some $L$ and\\ $\forall\x\in G_1(\x\in G_2\vee\x\in L)$.
 \end{proposition}
 \begin{proof}Straightforward.
 \end{proof}

 \begin{proposition}$ $\\
 \begin{enumerate}
   \item $\Gamma,\x\leqslant\Delta,\x$ iff $\,\Gamma\leqslant\Delta$;
   \item $\Gamma,\x\leqslant\Delta$ implies $\Delta$ has the form $\Sigma,\x$;
   \item $\Gamma\leqslant\Delta,\x$ implies $\Gamma$ has the form $\Sigma,\x$ or $\Gamma$ has the form $G$;
    \item $\Gamma\leqslant G$ implies $\Gamma$ is also a set;
    \item $G\leqslant\Gamma,\x$ iff $G-\{\x\}\leqslant\Gamma$
\end{enumerate}
 \end{proposition}
 \begin{proof}Straightforward.
 \end{proof}

  \begin{theorem}\label{increasse} If $\,\Gamma\vdash A$ is derivable and $\,\Gamma\leqslant\Sigma$ then $\Sigma\vdash A$ is derivable.
 \end{theorem}
 \begin{proof}Induction over the structure of $A$.\\[5pt]
 Case 1. $A$ is $\x$. Then $\Gamma\vdash\x$ is derivable iff $\x\in\Gamma$ (Lemma~\ref{lemmaax}). It is easy to prove that $\x\in\Gamma$ and $\Gamma\leqslant\Sigma$ imply $\x\in\Sigma$, hence $\Sigma\vdash\x$ is derivable.\\[5pt]
 Case 2. $A$ is $B_1B_2$. By Generation lemma $\Gamma\vdash B_1$ and $\Gamma\vdash B_2$ are derivable. By induction hypothesis $\Sigma\vdash B_1$ and $\Sigma\vdash B_2$ are derivable, hence $\Sigma\vdash B_1B_2$ is derivable.\\[5pt]
 Case 3. $A$ is $\lambda\x.B$. By Generation lemma $\Gamma,\x\vdash B$ is derivable. $\Gamma\leqslant\Sigma$ implies $\Gamma,\x\leqslant\Sigma,\x$. By induction hypothesis $\Sigma,\x\vdash B$ is derivable, hence $\Sigma\vdash\lambda\x.B$ is derivable.
  \end{proof}

  \begin{definition}$\Gamma$ and $\Delta$ are \emph{compatible} iff $\Gamma\leqslant\Sigma$ and $\Delta\leqslant\Sigma$ for some $\Sigma$.
 \end{definition}

 \begin{proposition}\label{cup} Any set $G$ is compatible with any $\Gamma$. \\
 $G_1,L_1$ and $G_2,L_2$ are compatible iff $L_1=LL_2$ or $L_2=LL_1$ for some $L$.\\
 If $\Gamma$ and $\Delta$ are compatible, there exists their supremum $\Gamma\sqcup\Delta$.
 \end{proposition}
 \begin{proof} $\Gamma\sqcup\Delta$, if exists, can be calculated recursively using the following rules:
 \begin{itemize}
   \item []$(\Gamma,\x)\sqcup(\Delta,\x)=(\Gamma\sqcup\Delta),\x$
   \item []$(\Gamma,\x)\sqcup G=(\Gamma\sqcup (G-\{\x\})),\x$
   \item []$G\sqcup(\Gamma,\x)=((G-\{\x\})\sqcup\Gamma),\x$
   \item []$G_1\sqcup G_2=G_1\cup G_2$
       \end{itemize}
    \end{proof}

    \begin{example}$(\{x\},z)\sqcup\{y,z\}=(\{x\}\cup\{y\}),z=\{x,y\},z$
    \end{example}

\begin{theorem}If $\Gamma\vdash A$ is derivable, then  $FV(A)\leqslant\Gamma$
\end{theorem}
\begin{proof} Induction over the structure of $A$.\\[5pt]
Case 1. $A$ is $\x$. Then $FV(A)=\{\x\}$. If $\Gamma\vdash\x$ is derivable, then $\x\in\Gamma$ (Lemma~\ref{lemmaax}), hence $\{\x\}\leqslant\Gamma$\\[5pt]
Case 2. $A$ is $B_1B_2$. Then $\Gamma\vdash B_1$ and $\Gamma\vdash B_2$ are derivable by Generation lemma. By induction hypothesis  $FV(B_1)\leqslant\Gamma$ and $FV(B_2)\leqslant\Gamma$. Hence   $FV(B_1B_2)=FV(B_1)\sqcup FV(B_2)\leqslant\Gamma$ \\[5pt]
Case 3. $A$ is $\lambda\x.B$. Then $\Gamma,\x\vdash B$ is derivable by Generation lemma. By induction hypothesis  $FV(B)\leqslant\Gamma,\x$. Hence\\
  $FV(\lambda\x.B)=FV(B)-\{\x\}\leqslant\Gamma$
   \end{proof}

   \begin{theorem} $FV(A)\vdash A$ is derivable for each $A$.
\end{theorem}
\begin{proof} Induction over the structure of $A$.\\[5pt]
Case 1. $A$ is $\x$. Then $FV(A)=\{\x\}$ and $\{\x\}\vdash\x$ is derivable.\\[5pt]
Case 2. $A$ is $B_1B_2$.  $FV(B_1B_2)=FV(B_1)\sqcup FV(B_2)$ \\
 $FV(B_1)\vdash B_1$ and $FV(B_2)\vdash B_2$ are derivable by  induction hypothesis.\\  $FV(B_1)\sqcup FV(B_2)\vdash B_1$ and $FV(B_1)\sqcup FV(B_2)\vdash B_2$ are derivable by Theorem~\ref{increasse}. Further
 $$\ruletwo{FV(B_1)\sqcup FV(B_2)\vdash B_1}{FV(B_1)\sqcup FV(B_2)\vdash B_2}{FV(B_1)\sqcup FV(B_2)\vdash B_1B_2}$$
   \\[5pt]
  Case 3. $A$ is $\lambda\x.B$. $FV(\lambda\x.B)=FV(B)-\{\x\}$\\
  $FV(B)\vdash B$ is derivable by induction hypothesis.
   Note that\\ $FV(B)<(FV(B)-\{\x\}),\x$\\
    Hence $(FV(B)-\{\x\}),\x\vdash B$ is derivable (Theorem~\ref{increasse}).
Further
 $$\ruleone{(FV(B)-\{\x\}),\x\vdash B}{FV(B)-\{\x\}\vdash\lambda\x.B}$$

\end{proof}

\begin{corollary} $FV(A)$ is the smallest context $\Gamma$ such that $\Gamma\vdash A$ is derivable.
\end{corollary}

\newpage

\section{Terms with explicit substitutions}
\begin{figure}
\begin{framed}
\noindent
\textbf{Syntax.} $\Lambda\alpha$ is the set of terms inductively defined by the following BNF:
\begin{align*}
\x::&= x\mid y\mid z\mid\ldots \tag{Variables}\\
A,B::&= \x \mid  AB \mid \lambda\x.A  \mid S\circ A \tag{Terms}\\
S::&=  [B/\x] \mid W_\x  \mid \{\y\x\} \mid\,\, \Up S_\x  \tag{Substitutions}
\end{align*}

\textbf{Inference rules.}
\begin{align*}
&R1 && G\vdash\x && (\x\in G)\\[5pt]
&R2&& \Gamma,\x\vdash \x &\\[5pt]
&R3&& \ruleone{\Gamma\vdash
\x}{\Gamma,\y\vdash \x} && (\x\neq \y) \\[10pt]
&R4&& \ruletwo{\Gamma\vdash A}{\Gamma\vdash B}{\Gamma\vdash
AB} &\\[10pt]
&R5&& \ruleone{\Gamma,\x\vdash A}{\Gamma\vdash\lambda \x.A} &\\[10pt]
&R6&& \ruletwo{\Gamma\vdash S\tri\Delta}{\Delta\vdash
A}{\Gamma\vdash
 S\circ A} &\\[10pt]
 &R7&& \ruleone{\Gamma\vdash
B}{\Gamma\vdash
[B/\x]\tri\Gamma ,\x} &\\[10pt]
 &R8&& \Gamma,\x\vdash W_\x \tri\Gamma &\\[5pt]
 &R9&& \Gamma,\y\vdash \{\y\x\} \tri\Gamma,\x &\\[5pt]
 &R10&& \ruleone{\Gamma\vdash S\tri\Delta}{\Gamma,\x\vdash\,\,
\Up S_\x\tri\Delta ,\x} &
\end{align*}
\end{framed}
\caption{Terms and substitutions}\label{inf1}
\end{figure}

\begin{definition}\emph{Terms} and \emph{substitutions}  are defined on Figure~\ref{inf1}.\\
A \emph{judgement} is an expression of the form $\Gamma\vdash A$
or of the form $\Gamma\vdash S\tri\Delta$.\\
 Inference rules for judgements are shown on Figure~\ref{inf1}.\\
  A term $A$ is \emph{well-formed} iff $\Gamma\vdash A$ is derivable for some $\Gamma$.\\
  Informally, \\
  $\begin{array}{lll}
S\circ A & \text{corresponds to}
 & A[S]\\
 W_\x & \text{corresponds to}
 & \uparrow\\
 \{\y\x\} & \text{corresponds to}
 & \un\,\cdot\!\uparrow\\
  \Up S_\x & \text{corresponds to}
 & \Up S
 \end{array}$
\end{definition}

\begin{convention}$ $\\[5pt]
$\begin{array}{lll}
S\circ AB_1\ldots B_k & \text{is shorthand for}
 & S\circ(AB_1\ldots B_k)\\
 \lambda \x.S\circ A & \text{is shorthand for}
 & \lambda \x.(S\circ A)\\
 S\circ \lambda \x.A & \text{is shorthand for}
 & S\circ(\lambda \x.A)\\
  S_1\circ \ldots\circ S_k\circ A & \text{is shorthand for}
 & S_1\circ(\ldots\circ(S_{k-1}\circ(S_k\circ A))\ldots)\\
\Up\ldots\Up S_{\x\y\ldots\z} & \text{is shorthand for}
 & \Up(\ldots(\Up(\Up S_{\x})_{\y})\ldots)_{\z}
 \end{array}$
\end{convention}

\begin{example}$ $\\[5pt]
$\begin{array}{lll}
W_x\circ W_y\circ\lambda z.W_z\circ xy & \text{is shorthand for} & W_x\circ (W_y\circ(\lambda z.(W_z\circ (xy))))
\end{array}$
\end{example}

\begin{example}\label{strange}
$$\ruleone{\ruletwo{\{x\},x\vdash W_x\tri\{x\}}{\{x\}\vdash x}{\{x\},x\vdash W_x\circ x}}{\{x\}\vdash\lambda x.W_x\circ x}$$
\end{example}

\begin{proposition}The following rules are admissible:\\[10pt]
$$\ruleone
{\Gamma\vdash A}
{\Gamma,\x\vdash W_\x\circ A}\qquad
\ruleone
{\Gamma,\x\vdash A}
{\Gamma,\y\vdash\{\y\x\}\circ A}\qquad
\ruletwo
{\Gamma,\x\vdash A}
{\Gamma\vdash B}
{\Gamma\vdash[B/\x]\circ A}$$
\end{proposition}
\begin{proof}
$$\ruletwo
{\Gamma,\x\vdash W_\x\tri\Gamma}
{\ruledot{\Gamma\vdash A}}
{\Gamma,\x\vdash W_\x\circ A}\qquad
\ruletwo
{\Gamma,\y\vdash\{\y\x\}\tri\Gamma,\x}
{\ruledot{\Gamma,\x\vdash A}}
{\Gamma,\y\vdash\{\y\x\}\circ A}\qquad
\ruletwo
{\ruleone
{\ruledot{\Gamma\vdash B}}
{\Gamma\vdash[B/\x]\tri\Gamma,\x}}
{\ruledot{\Gamma,\x\vdash A}}
{\Gamma\vdash[B/\x]\circ A}$$

\end{proof}

\begin{example}\label{notwellformterm} A
judgement of the form $\Gamma\vdash\lambda \x.W_\y\circ A $ is not derivable
 if $\x\neq \y$.
$$
\ruleone{\ruletwo{\Gamma,\y\vdash
W_\y\tri\Gamma}{\ruledot{\Gamma\vdash A}}{\Gamma,\y\vdash W_\y\circ
A }}{\Gamma\vdash\lambda \x.W_\y\circ A \using(?)}
$$
Hence, a term of the form $\lambda \x.W_{\y}\circ A $ is not well-formed
 if $\x\neq \y$.
\end{example}

\begin{example}A term of the form $(W_\x\circ A)(W_\y\circ B)$ is not
well-formed  if $\x\neq \y$.
$$
\ruletwo{\ruletwo{\Gamma,\x\vdash W_\x\tri\Gamma}{\ruledot{\Gamma\vdash
A}}{\Gamma,\x\vdash W_\x\circ A}}
{\ruletwo{\Delta,\y\vdash W_\y\tri\Delta}{\ruledot{\Delta\vdash
B}}{\Delta,\y\vdash W_\y\circ B}} {?\vdash(W_\x\circ A)(W_\y\circ
B)}
$$
\end{example}

\begin{lemma}[Generation lemma]$ $\\
Each derivation of $\,G\vdash \x$ is an application of the
rule~$R1$.\\
Each derivation of $\,\Gamma,\x\vdash \x$ is an application of the
rule~$R2$.\\
Each derivation of $\,\Gamma,\y\vdash \x$ $($where $\x\neq \y)$ is
an application of the rule~$R3$ to some derivation of
$\,\Gamma\vdash
\x$.\\
Each derivation of $\,\Gamma\vdash AB$ is an application of the
rule~$R4$ to some  derivations of $\,\Gamma\vdash A$ and
$\,\Gamma\vdash B$.\\
Each derivation of $\,\Gamma\vdash \lambda \x.A$  is an
application of the rule~$R5$ to some derivation of
$\,\Gamma,\x\vdash A$.\\
Each derivation of $\,\Gamma\vdash S\circ A$ is an application of
the rule~$R6$ to some derivations of $\,\Gamma\vdash S\tri\Delta$
and $\,\Delta\vdash A$ for some
$\Delta$.\\
Each derivation of $\,\Gamma\vdash [B/\x]\tri\Delta$ is an application
of the
rule~$R7$ to some derivation of $\,\Gamma\vdash B$, where $\Delta$ is $\Gamma,\x$.\\
Each derivation of $\,\Delta\vdash W_\x\tri\Gamma$ is an
application of the rule~$R8$, where $\Delta$ is $\Gamma,\x$.\\
Each derivation of $\,\Delta\vdash\{\y\x\}\tri\Sigma$
 is an application of the rule~$R9$, where  $\Delta$ is $\Gamma,\y$ and $\Sigma$ is $\Gamma,\x$ for some $\Gamma$.\\
Each derivation of $\,\Sigma\vdash \,\,\Up S_\x\tri\Psi$ is an
application of the rule~$R10$ to some derivation of
$\,\Gamma\vdash S\tri\Delta$, where $\Sigma$ is $\Gamma,\x$ and $\Psi$ is $\Delta,\x$.
\end{lemma}
\begin{proof}
The proof is straightforward.\\
\end{proof}

\begin{corollary}Subterms of well-formed terms are well-formed.
\end{corollary}

\begin{lemma}\label{lemmax}$\Gamma\vdash\x$ is derivable iff $\x\in\Gamma$.
 \end{lemma}
 \begin{proof}Induction over the length of local part of $\Gamma$. If this length is equal to $0$, then $\Gamma$ has the form $G$ and $G\vdash\x$ is derivable iff $\x\in G$. If $\Gamma$ has the form $\Delta,\y$ then either $\x=\y$  or $\x\neq\y$. In the first case $\Gamma\vdash\x$ is derivable and $\x\in\Gamma$. In the last case $\Delta,\y\vdash\x$ is derivable iff $\Delta\vdash\x$ is derivable and we use the induction hypothesis.\\
 \end{proof}

\begin{proposition}
 If a judgement of the form $\Gamma\vdash S\tri\Delta$ is
derivable, then  $\Delta$ is uniquely determined by $\Gamma$
and $S$.
\end{proposition}
\begin{proof}The proof is by induction over the structure of
$S$.\\
Case 1. $S$ has the form $[B/\x]$. Then $\Delta$ is
$\Gamma,\x$.\\
Case 2. $S$ has the form $W_\x$. Then
 $\Gamma$
is $\Delta,\x$.\\
Case 3. $S$ has the form $\{\y\x\}$. Then $\Gamma$ has the form $\Sigma,\y$ and $\Delta$ is $\Sigma,\x$.\\
 Case 4. $S$ has the form $\Up S'_\x$. By Generation
lemma, we can derive $\Sigma\vdash S'\tri\Psi$, where $\Gamma$ is $\Sigma,\x$ and $\Delta$ is $\Psi,\x$.
By the induction hypothesis, $\Psi$ is uniquely determined by $\Sigma$ and $S'$.\\
\end{proof}

\begin{proposition}\label{unique}For any derivable judgement, its
derivation is unique.
\end{proposition}
\begin{proof} We  construct the derivation from the bottom up, using Generation lemma and the previous proposition.\\
\end{proof}

 \begin{theorem}\label{increase}If $\,\Gamma\vdash A$ is derivable and $\,\Gamma\leqslant\Sigma$ then $\Sigma\vdash A$ is derivable. If $\,\Gamma\vdash S\tri\Delta$ is derivable and $\,\Gamma\leqslant\Sigma$ then $\Sigma\vdash S\tri\Psi$ is derivable for some $\Psi\geqslant\Delta$.
 \end{theorem}
 \begin{proof}Induction over the structure of $A$ and $S$.\\[5pt]
 Case 1. $A$ is $\x$. Then $\Gamma\vdash\x$ is derivable iff $\x\in\Gamma$ (Lemma~\ref{lemmax}). It is easy to prove that $\x\in\Gamma$ and $\Gamma\leqslant\Sigma$ imply $\x\in\Sigma$, hence $\Sigma\vdash\x$ is derivable.\\[5pt]
 Case 2. $A$ is $B_1B_2$. By Generation lemma $\Gamma\vdash B_1$ and $\Gamma\vdash B_2$ are derivable. By induction hypothesis $\Sigma\vdash B_1$ and $\Sigma\vdash B_2$ are derivable, hence $\Sigma\vdash B_1B_2$ is derivable.\\[5pt]
 Case 3. $A$ is $\lambda\x.B$. By Generation lemma $\Gamma,\x\vdash B$ is derivable. $\Gamma\leqslant\Sigma$ implies $\Gamma,\x\leqslant\Sigma,\x$. By induction hypothesis $\Sigma,\x\vdash B$ is derivable, hence $\Sigma\vdash\lambda\x.B$ is derivable.\\[5pt]
 Case 4. $A$ is $S\circ B$. By Generation lemma $\Gamma\vdash S\tri\Delta$ and $\Delta\vdash B$ are derivable for some $\Delta$. By induction hypothesis $\Sigma\vdash S\tri\Psi$ and $\Psi\vdash B$ are derivable for some $\Psi\geqslant\Delta$, hence $\Sigma\vdash S\circ B$ is derivable.\\[5pt]
 Case 5. $S$ is $[B/\x]$. $\Gamma\vdash[B/\x]\tri\Gamma,\x$ is derivable, hence $\Gamma\vdash B$ is derivable. By induction hypothesis $\Sigma\vdash B$ is derivable, hence $\Sigma\vdash[B/\x]\tri\Sigma,\x$ is derivable. $\Gamma\leqslant\Sigma$ implies $\Gamma,\x\leqslant\Sigma,\x$.\\[5pt]
 Case 6. $S$ is $W_\x$. Then $\Gamma$ has the form $\Delta,\x$. $\Delta,\x\leqslant\Sigma$ implies $\Sigma$ has the form $\Psi,\x$. $\Psi,\x\vdash W_\x\tri\Psi$ is derivable. $\Gamma\leqslant\Sigma$ implies $\Delta\leqslant\Psi$.\\[5pt]
 Case 7. $S$ is $\{\y\x\}$. Then $\Gamma$ has the form $\Gamma',\y$ and $\Delta$ has the form $\Gamma',\x$. $\Gamma',\y\leqslant\Sigma$ implies $\Sigma$ has the form $\Sigma',\y$. Put $\Psi=\Sigma',\x$, then $\Sigma\vdash\{\y\x\}\tri\Psi$ is derivable and $\Delta\leqslant\Psi$.\\[5pt]
 Case 8. $S$ is $\Up S'_\x$. Then $\Gamma$ has the form $\Gamma',\x$ and $\Delta$ has the form $\Delta',\x$ where $\Gamma'\vdash S'\tri\Delta'$ is derivable. $\Gamma',\x\leqslant\Sigma$ implies $\Sigma$ has the form $\Sigma',\x$ and $\Gamma'\leqslant\Sigma'$. By induction hypothesis $\Sigma'\vdash S'\tri\Psi'$ is derivable for some $\Psi'\geqslant\Delta'$. Put $\Psi=\Psi',\x$. Then $\Sigma\vdash\,\,\Up S'_\x\tri\Psi$ is derivable and $\Psi\geqslant\Delta$.\\
 \end{proof}

 \newpage

\section{Free variables}
\begin{figure}
\begin{framed}
\noindent
\begin{align*}
FV(\x)&=\{\x\}\\
FV(AB)&=FV(A)\sqcup FV(B)\\
FV(\lambda\x.A)&= O_{\lambda\x}(FV(A))\\
FV(W_\x\circ A)&=FV(A),\x\\
FV([B/\x]\circ A])&=FV ((\lambda\x.A)B)\\
FV(\{\y\x\}\circ A)&=FV(W_\y\circ\lambda\x.A)\\
FV(\Up S_\x\circ A)&= FV(W_\x\circ S\circ\lambda\x.A)
\end{align*}

\begin{align*}
O_{\lambda\x}(\Gamma,\x)&=\Gamma\\
O_{\lambda\x}(G)&=G-\{\x\}
\end{align*}
\end{framed}
\caption{Free variables}\label{elim}
\end{figure}

\begin{definition}The definition of free variables is shown on Figure~\ref{elim}. $FV(A)$ is  not a set, but a context (the smallest context $\Gamma$ such that $\Gamma\vdash A$ is derivable). \end{definition}

\begin{example}$FV(\lambda x.xy)= O_{\lambda x}(FV(xy))= O_{\lambda x}(FV(x)\sqcup FV(y))\\= O_{\lambda x} (\{\x\}\cup\{\y\})=O_{\lambda x}(\{x,y\})=\{y\}$
 \end{example}

 \begin{example}$FV(W_x\circ z)=FV(z),x=\{z\},x$
 \end{example}

 \begin{example}$FV(\lambda x.W_x\circ x)= O_{\lambda x}(FV(W_x\circ x))=O_{\lambda x}(FV(x),x)\\=O_{\lambda x}(\{x\},x)=\{x\}$
 \end{example}

  Note that $x\in FV(\lambda x.W_x\circ x)$.  It is possible now that $\x\in FV(\lambda\x.A)$\\ See Example~\ref{strange} to understand.

  Note that $FV(A)$ does not always exist. For example,\\ $FV(\lambda\x.W_\y\circ A)= O_{\lambda\x}(FV(W_\y\circ A))=O_{\lambda\x}(FV(A),\y)$ does not exist if $\x\neq\y$ (but such term is not well-formed).

  Note that

  \begin{itemize}
  \item[] $O_{\lambda\x}(\Delta)\leqslant\Sigma$ iff $\Delta\leqslant\Sigma,\x$
  \end{itemize}

 \begin{lemma}\label{lamcup}If $O_{\lambda\x}(\Gamma\sqcup\Delta)$ exists, then \begin{itemize}
 \item[] $O_{\lambda\x}(\Gamma\sqcup\Delta)=O_{\lambda\x}(\Gamma)\sqcup O_{\lambda\x}(\Delta)$
 \end{itemize}
\end{lemma}

\begin{proof}See the proof of Proposition~\ref{cup}.\\
\end{proof}

\begin{corollary}If $O_{\lambda\x}(\Gamma)$ exists and $\Gamma\geqslant\Delta$, then
 $O_{\lambda\x}(\Delta)$ exists\\ and $O_{\lambda\x}(\Gamma)\geqslant O_{\lambda\x}(\Delta)$
     \end{corollary}

\begin{lemma}\label{lemma1}$\Gamma\vdash[B/\x]\circ A$ is derivable iff $\,\Gamma\vdash(\lambda\x.A)B$ is derivable.
 \end{lemma}

 \begin{proof}$ $\\
  \ruletwo{\ruleone{\ruledot{\Gamma\vdash B}}{\Gamma\vdash[B/\x]\tri\Gamma,\x}}{\ruledot{\Gamma,\x\vdash A}}{\Gamma\vdash[B/\x]\circ A}
  \qquad\qquad\ruletwo{\ruleone{\ruledot{\Gamma,\x\vdash A}}{\Gamma\vdash\lambda\x.A}}{\ruledot{\Gamma\vdash B}}{\Gamma\vdash(\lambda\x.A)B}
  \end{proof}

  \begin{lemma}\label{lemma2}$\Gamma\vdash\{\y\x\}\circ A$ is derivable iff $\,\Gamma\vdash W_\y\circ\lambda\x.A$ is derivable.
  \end{lemma}
  \begin{proof}$ $\\
  \ruletwo{\Delta,\y\vdash\{\y\x\}\tri\Delta,\x}{\ruledot{\Delta,\x\vdash A}}{\Delta,\y\vdash\{\y\x\}\circ A}\\[20pt]

  \ruletwo{\Delta,\y\vdash W_\y\tri\Delta}
  {\ruleone{\ruledot{\Delta,\x\vdash A}}{\Delta\vdash\lambda\x.A}}
  {\Delta,\y\vdash W_\y\circ\lambda\x.A}

  \end{proof}

  \begin{lemma}\label{lemma3}$\Gamma\vdash\,\,\Up S_\x\circ A$ is derivable iff $\,\Gamma\vdash  W_\x\circ S\circ\lambda\x.A$ is derivable.
  \end{lemma}
  \begin{proof}$ $\\
  \ruletwo{\ruleone{\ruledot{\Delta\vdash S\tri\Sigma}}{\Delta,\x\vdash\,\,\Up S_\x\tri\Sigma,\x}}
  {\ruledot{\Sigma,\x\vdash A}}
  {\Delta,\x\vdash\,\,\Up S_\x\circ A}

  \ruletwo{\Delta,\x\vdash W_\x\tri\Delta}
  {\ruletwo{\ruledot{\Delta\vdash S\tri\Sigma}}
  {\ruleone{\ruledot{\Sigma,\x\vdash A}}{\Sigma\vdash\lambda\x.A}}
  {\Delta\vdash S\circ\lambda\x.A}}
  {\Delta,\x\vdash W_\x\circ S\circ\lambda\x.A}

\end{proof}
 \begin{theorem}If $\Gamma\vdash A$ is derivable, then $FV(A)$ exists and $FV(A)\leqslant\Gamma$.
\end{theorem}
\begin{proof}Induction over

 \begin{enumerate}
   \item the total number of  $[\,]$,  $\{\}$, and $\Up\,$ in $A$;
   \item the length of $A$.
 \end{enumerate}
 There are seven possible cases:\\[5pt]
 Case 1. $A$ is $\x$. Then $FV(A)=\{\x\}$. If $\Gamma\vdash\x$ is derivable, then $\x\in\Gamma$ (Lemma~\ref{lemmax}), hence $\{\x\}\leqslant\Gamma$.\\[5pt]
Case 2. $A$ is $B_1B_2$. Then $\Gamma\vdash B_1$ and $\Gamma\vdash B_2$ are derivable by Generation lemma. By induction hypothesis  $FV(B_1)\leqslant\Gamma$ and $FV(B_2)\leqslant\Gamma$. Hence $FV(B_1)\sqcup FV(B_2)$ exists (Proposition~\ref{cup}) and  $FV(B_1)\sqcup FV(B_2)\leqslant\Gamma$.\\[5pt]
Case 3. $A$ is $\lambda\x.B$. Then $\Gamma,\x\vdash B$ is derivable by Generation lemma. By induction hypothesis $FV(B)\leqslant\Gamma,\x$. Hence $FV(B)$ has the form $\Delta,\x$ or $FV(B)$ has the form $G$.
In both cases $FV(\lambda\x.B)$ exists and\\ $FV(\lambda\x.B)=O_{\lambda\x}(FV(B))\leqslant O_{\lambda\x}(\Gamma,\x)=\Gamma$.\\[5pt]
Case 4. $A$ is $W_\x\circ B$. Then $\Gamma$ has the form $\Delta,\x$ and $\Delta\vdash B$ is derivable by Generation lemma. By induction hypothesis   $FV(B)\leqslant\Delta$, hence\\  $FV(B),\x\leqslant\Delta,\x$. Further, $FV(W_\x\circ B)=FV(B),\x$, hence $FV(W_\x\circ B)\leqslant\Gamma$.\\[5pt]
Case 5. $A$ has the form $[B_1/\x]\circ B_2$. Use Lemma~\ref{lemma1}.\\[5pt]
Case 6. $A$ has the form $\{\y\x\}\circ B$. Use Lemma~\ref{lemma2}.\\[5pt]
Case 7. $A$ has the form $\Up S_\x\circ B$. Use Lemma~\ref{lemma3}.\\
\end{proof}

\begin{theorem} If $FV(A)$ exists, then $FV(A)\vdash A$ is derivable.
\end{theorem}

\begin{proof}Induction over

 \begin{enumerate}
   \item the total number of  $[\,]$,  $\{\}$, and $\Up\,$ in $A$;
   \item the length of $A$.
 \end{enumerate}
 There are seven possible cases:\\[5pt]
 Case 1. $A$ is $\x$. Then $FV(A)=\{\x\}$ and $\{\x\}\vdash\x$ is derivable.\\[5pt]
Case 2. $A$ is $B_1B_2$. $FV(B_1B_2)=FV(B_1)\sqcup FV(B_2)$. \\
 $FV(B_1)\vdash B_1$ and $FV(B_2)\vdash B_2$ are derivable by  induction hypothesis.\\  $FV(B_1)\sqcup FV(B_2)\vdash B_1$ and $FV(B_1)\sqcup FV(B_2)\vdash B_2$ are derivable by Theorem~\ref{increase}.\\
   Hence $FV(B_1)\sqcup FV(B_2)\vdash B_1B_2$ is derivable.\\[5pt]
Case 3. $A$ is $\lambda\x.B$. $FV(\lambda\x.B)=O_{\lambda\x}(FV(B))$.  Hence $FV(B)$ has the form $\Delta,\x$ or $FV(B)$ has the form $G$. By induction hypothesis $FV(B)\vdash B$ is derivable.\\
If $FV(B)=\Delta,\x$ then $\Delta,\x\vdash B$, hence $\Delta\vdash\lambda\x.B$\\ Note that  $FV(\lambda\x.B)=\Delta$, hence $FV(\lambda\x.B)\vdash\lambda\x.B$\\
If $FV(B)=G$ then $G\vdash B$ is derivable.\\ Note that $G<(G-\{\x\}),\x$. Hence $(G-\{\x\}),\x\vdash B$ is derivable (Theorem~\ref{increase})\\
Further $G-\{\x\}\vdash\lambda\x.B$ is derivable and $FV(\lambda\x.B)=G-\{\x\}$\\[5pt]
Case 4. $A$ is $W_\x\circ B$. $FV(W_\x\circ B)=FV(B),\x$. By induction hypothesis   $FV(B)\vdash B$ is derivable, hence  $FV(B),\x\vdash W_\x\circ B$ is derivable.\\[5pt]
Case 5. $A$ has the form $[B_1/\x]\circ B_2$. Use Lemma~\ref{lemma1}.\\[5pt]
Case 6. $A$ has the form $\{\y\x\}\circ B$. Use Lemma~\ref{lemma2}.\\[5pt]
Case 7. $A$ has the form $\Up S_\x\circ B$. Use Lemma~\ref{lemma3}.\\
\end{proof}

\begin{corollary}\label{FreV} $FV(A)$ exists iff $A$ is well-formed. In this case $FV(A)$ is the smallest context $\Gamma$ such that $\Gamma\vdash A$ is derivable.
\end{corollary}

\newpage

\section{The calculus $\lambda\alpha$}
\begin{figure}
\begin{framed}
\noindent
\begin{align*}
&(Beta) &  (\lambda \x.A)B &\to [B/\x]\circ A  \\
&(App) & S\circ AB &\to (S\circ A)(S\circ B) \\
&(Lambda) & S\circ\lambda\x. A &\to\lambda\x.\!\Up S_\x\circ A \\
&(Var) & [ B/\x]\circ \x &\to B \\
&(Shift) & [B/\x]\circ W_\x\circ A &\to A \\
&(Shift') & [ B/\x]\circ\z &\to \z \tag{$\x\neq \z$}\\
&(IdVar) & \{\y \x\}\circ \x &\to \y \\
&(IdShift) & \{\y\x\}\circ W_\x\circ A &\to W_\y\circ A \\
&(IdShift') & \{\y \x\} \circ\z &\to W_\y\circ  \z \tag{$\x\neq \z$}\\
&(LiftVar) & \Up S_\x\circ \x &\to \x \\
&(LiftShift) & \Up S_\x\circ W_\x\circ A &\to W_\x\circ S\circ A \\
&(LiftShift') & \Up S_\x \circ\z &\to W_\x\circ S\circ \z \tag{$\x\neq \z$}\\
&(W) & W_\x\circ\z &\to \z \tag{$\x\neq \z$}\\
&(\alpha) & \lambda\x. A &\to \lambda \y.\{\y \x\} \circ A \tag{$*$}
\end{align*}

Here $(*)$ is the condition $\x\in FV(\lambda\x.A)\,\&\, \y\not\in FV(\lambda\x.A)$\\

\end{framed}
\caption{The calculus $\lambda\alpha$}\label{lamalp}
\end{figure}

\begin{figure}
\begin{framed}
\noindent
$\begin{array}{ccc}
& \ruleone{ A\to A'}{ \lambda
\x.A\to\lambda \x.A'}\\[20pt]
 \ruleone{ A\to A'}{ AB\to  A'B} && \ruleone{
B\to B'}{
AB\to  AB'}\\[20pt]
 \ruleone{  S\to S'}{S\circ A\to S'\circ A} && \ruleone{
A\to
A'}{S\circ A\to S\circ A'}\\[20pt]
& \ruleone{ B\to B'}{ [B/\x]\to[B'/\x]}\\[20pt]
& \ruleone{ S\to S'}{ \Up S_\x\to\,\, \Up S'_\x}
\end{array}$
\end{framed}
\caption{Compatible closure}\label{closure}
\end{figure}

\begin{definition}The calculus $\lambda\alpha$ is shown on  Figure~\ref{lamalp} and Figure~\ref{closure}.
\end{definition}

The meaning of the rule $W$ is as follows: if $\Gamma,\x\vdash W_\x\circ\z$ is derivable, then $W_\x\circ\z$ denote the rightmost $\z$ in $\Gamma$.

 $$\ruleone{\Gamma\vdash\z}{\Gamma,\x\vdash W_{\x}\circ\z}$$
 But if $\x\neq\z$, the rightmost $\z$ in $\Gamma$ is the same as the rightmost $\z$ in $\Gamma,\x$. Hence $\Gamma,\x\vdash\z$ is the same as $\Gamma,\x\vdash W_\x\circ\z$. The idea is not new, see \cite{Granstrom} for example.
\newpage
The rules $Shift'$, $IdShift'$, and $LiftShift'$ provide confluence in the following cases:

$$\begin{tikzcd}
 [B/\x]\circ W_\x\circ\z  \arrow{rr}{Shift} \arrow{dd}[swap]{W} && \z \\
  & & \\
 {[} B/\x]\circ \z  &&
 \end{tikzcd}$$

$$\begin{tikzcd}
 \{\y\x\}\circ W_\x\circ\z  \arrow{rr}{IdShift} \arrow{dd}[swap]{W} && W_\y\circ\z \\
  & & \\
   \{\y\x\}\circ\z  &&
 \end{tikzcd}$$

$$\begin{tikzcd}
 \Up S_{\x}\circ W_\x\circ\z  \arrow{rr}{LiftShift} \arrow{dd}[swap]{W} && W_\x\circ S\circ\z \\
  & & \\
   \Up S_{\x}\circ\z  &&
 \end{tikzcd}$$

\begin{example}
\begin{align*}
&(\lambda xy.x)\,y && \\
&\to [y/x]\circ\lambda y.x && (Beta)\\
&\to \lambda y.\!\Up[y/x]_y\circ x && (Lambda)\\
&\to \lambda y.W_y\circ[y/x]\circ x && (LiftShift')\\
&\to \lambda y.W_y\circ y && (Var)\\
&\to \lambda z.\{zy\}\circ W_y\circ y && (\alpha)\\
&\to \lambda z.W_z\circ y && (IdShift)\\
&\to \lambda z.y && (W)
\end{align*}
\end{example}

\begin{example}
\begin{align*}
&(\lambda xyz.xz(yz))(\lambda xy.x) && \\
&\to [\lambda xy.x/x]\circ\lambda yz.xz(yz) && (Beta)\\
&\to \lambda y.\!\Up[\lambda xy.x/x]_y\circ \lambda z.xz(yz) && (Lambda)\\
&\to \lambda yz.\!\Up\,\Up[\lambda xy.x/x]_{yz}\circ xz(yz) && (Lambda)\\
&\to \lambda yz.(\underline{\Up\,\Up[\lambda xy.x/x]_{yz}\circ xz})(\Up\,\Up[\lambda xy.x/x]_{yz}\circ yz) && (App)\\
&\twoheadrightarrow \lambda yz.(\lambda y.z)(\underline{\Up\,\Up[\lambda xy.x/x]_{yz}\circ yz}) && (Example~\ref{example104})\\
&\twoheadrightarrow \lambda yz.(\lambda y.z)(yz) && (Example~\ref{example105})\\
&\to \lambda yz.[yz/y]\circ z && (Beta)\\
&\to \lambda yz.z && (Shift')
\end{align*}
\end{example}

\begin{example}\label{example104}
\begin{align*}
&\Up\,\Up[\lambda xy.x/x]_{yz}\circ xz && \\
&\to  (\Up\,\Up[\lambda xy.x/x]_{yz}\circ x)(\underline{\Up\,\Up[\lambda xy.x/x]_{yz}\circ z}) && (App)\\
&\to (\Up\,\Up[\lambda xy.x/x]_{yz}\circ x)\, z && (LiftVar)\\
&\to (W_z\circ\Up[\lambda xy.x/x]_y\circ x)\, z && (LiftShift')\\
&\to (W_z\circ W_y\circ[\lambda xy.x/x]\circ x)\, z && (LiftShift')\\
&\to (W_z\circ W_y\circ\lambda xy.x)\, z && (Var)\\
&\twoheadrightarrow (\lambda xy.x)\, z && (\text{cause\,\,} \lambda xy.x \text{\,\,is closed})\\
&\to [z/x]\circ\lambda y.x && (Beta)\\
&\to \lambda y.\!\Up[z/x]_y\circ x && (Lambda)\\
&\to \lambda y.W_y\circ[z/x]\circ x && (LiftShift')\\
&\to \lambda y.W_y\circ z && (Var)\\
&\to \lambda y.z && (W)
\end{align*}
\end{example}

\begin{example}\label{example105}
\begin{align*}
&\Up\,\Up[\lambda xy.x/x]_{yz}\circ yz && \\
&\to  (\Up\,\Up[\lambda xy.x/x]_{yz}\circ y)(\underline{\Up\,\Up[\lambda xy.x/x]_{yz}\circ z}) && (App)\\
&\to (\Up\,\Up[\lambda xy.x/x]_{yz}\circ y)\, z && (LiftVar)\\
&\to (W_z\,\circ\Up[\lambda xy.x/x]_y\circ y)\, z && (LiftShift')\\
&\to  (W_z\circ y)\, z && (LiftVar)\\
&\to yz && (W)
\end{align*}
\end{example}

\begin{theorem}\label{subbred} ``Subject reduction''.\\
If $\Gamma\vdash A$ and $A\to B$ then  $\Gamma\vdash B$.
\end{theorem}
\begin{proof}$ $\\
Case $Beta$.\\[10pt]
\ruletwo{\ruleone{\ruledot{\Gamma,\x\vdash A}}{\Gamma\vdash\lambda\x.A}}{\ruledot{\Gamma\vdash B}}{\Gamma\vdash(\lambda\x.A)B}\qquad\qquad\ruletwo{\ruleone{\ruledot{\Gamma\vdash B}}{\Gamma\vdash[B/\x]\tri\Gamma,\x}}{\ruledot{\Gamma,\x\vdash A}}{\Gamma\vdash[B/\x]\circ A}
\newpage
\noindent
Case $App$.\\[10pt]
\ruletwo{\ruledot{\Gamma\vdash S\tri\Delta}}{\ruletwo{\ruledot{\Delta\vdash A}}{\ruledot{\Delta\vdash B}}{\Delta\vdash AB}}{\Gamma\vdash S\circ AB}\qquad\quad
\ruletwo
{\ruletwo{\ruledot{\Gamma\vdash S\tri\Delta}}{\ruledot{\Delta\vdash A}}{\Gamma\vdash S\circ A}}
{\ruletwo{\ruledot{\Gamma\vdash S\tri\Delta}}{\ruledot{\Delta\vdash B}}{\Gamma\vdash S\circ B}}
{\Gamma\vdash(S\circ A)(S\circ B)}\\[20pt]
Case $Lambda$.\\[10pt]
   \ruletwo{\ruledot{\Gamma\vdash S\tri\Delta}}
  {\ruleone{\ruledot{\Delta,\x\vdash A}}{\Delta\vdash\lambda\x.A}}
  {\Gamma\vdash S\circ\lambda\x.A}\qquad\qquad
    \ruleone{\ruletwo{\ruleone{\ruledot{\Gamma\vdash S\tri\Delta}}{\Gamma,\x\vdash \,\,\Up S_\x\tri\Delta,\x}}
  {\ruledot{\Delta,\x\vdash A}}
  {\Gamma,\x\vdash \,\,\Up S_\x\circ A}}
  {\Gamma\vdash\lambda\x.\!\Up S_\x\circ A}\\[20pt]
  Case $Var$.\\[10pt]
  \ruletwo
 {\ruleone{\ruledot{\Gamma\vdash B}}{\Gamma\vdash[B/\x]\tri\Gamma,\x}}
 {\Gamma,\x\vdash \x}
 {\Gamma\vdash[B/\x]\circ\x}\qquad\qquad
  \ruledot{\Gamma\vdash B}\\[20pt]
  Case $Shift$.\\[10pt]
\ruletwo
{\ruleone{\ruledot{\Gamma\vdash B}}
{\Gamma\vdash[B/\x]\tri\Gamma,\x}}
{\ruletwo
{\Gamma,\x\vdash W_\x\tri\Gamma}
{\ruledot{\Gamma\vdash A}}
{\Gamma,\x\vdash W_\x\circ A}}
{\Gamma\vdash[B/\x]\circ W_\x\circ A}\qquad\qquad
\ruledot{\Gamma\vdash A}\\[20pt]
Case $Shift'$.\\[10pt]
\ruletwo
 {\ruleone{\ruledot{\Gamma\vdash B}}{\Gamma\vdash[B/\x]\tri\Gamma,\x}}
 {\ruleone{\ruledot{\Gamma\vdash \z}}{\Gamma,\x\vdash\z}}
 {\Gamma\vdash[B/\x]\circ\z}\qquad\qquad
  \ruledot{\Gamma\vdash \z}\\[20pt]
Case $IdVar$.\\[10pt]
\ruletwo
{\Gamma,\y\vdash\{\y\x\}\tri\Gamma,\x}
{\Gamma,\x\vdash\x}
{\Gamma,\y\vdash\{\y\x\}\circ\x}\qquad\qquad
$\Gamma,\y\vdash\y$\\[20pt]
Case $IdShift$.\\[10pt]
\ruletwo
{\Gamma,\y\vdash\{\y\x\}\tri\Gamma,\x}
{\ruletwo
{\Gamma,\x\vdash W_\x\tri\Gamma}
{\ruledot{\Gamma\vdash A}}
{\Gamma,\x\vdash W_\x\circ A}}
{\Gamma,\y\vdash\{\y\x\}\circ W_\x\circ A}\qquad\quad
\ruletwo
{\Gamma,\y\vdash W_\y\tri\Gamma}
{\ruledot{\Gamma\vdash A}}
{\Gamma,\y\vdash W_\y\circ A}\\[20pt]
Case $IdShift'$.\\[10pt]
\ruletwo
{\Gamma,\y\vdash\{\y\x\}\tri\Gamma,\x}
{\ruleone{\ruledot{\Gamma\vdash\z}}
{\Gamma,\x\vdash\z}}
{\Gamma,\y\vdash\{\y\x\}\circ\z}\qquad\qquad
\ruletwo
{\Gamma,\y\vdash W_\y\tri\Gamma}
{\ruledot{\Gamma\vdash\z}}
{\Gamma,\y\vdash W_\y\circ\z}\\[20pt]
Case $LiftVar$.\\[10pt]
\ruletwo
{\ruleone{\ruledot{\Gamma\vdash S\tri\Delta}}
{\Gamma,\x\vdash \,\,\Up S_\x\tri\Delta,\x}}
{\Delta,\x\vdash\x}
{\Gamma,\x\vdash \,\,\Up S_\x\circ\x}\qquad\qquad
$\Gamma,\x\vdash\x$\\[20pt]
Case $LiftShift$.\\[10pt]
\ruletwo
{\ruleone{\ruledot{\Gamma\vdash S\tri\Delta}}
{\Gamma,\x\vdash \,\,\Up S_\x\tri\Delta,\x}}
{\ruletwo
{\Delta,\x\vdash W_\x\tri\Delta}
{\ruledot{\Delta\vdash A}}
{\Delta,\x\vdash W_\x\circ A}}
{\Gamma,\x\vdash \,\,\Up S_\x\circ W_\x\circ A}\\[20pt]
\ruletwo
{\Gamma,\x\vdash W_\x\tri\Gamma}
{\ruletwo
{\ruledot{\Gamma\vdash S\tri\Delta}}
{\ruledot{\Delta\vdash A}}
{\Gamma\vdash S\circ A}}
{\Gamma,\x\vdash W_\x\circ S\circ A}\\[20pt]
Case $LiftShift'$.\\[10pt]
\ruletwo
{\ruleone{\ruledot{\Gamma\vdash S\tri\Delta}}
{\Gamma,\x\vdash \,\,\Up S_\x\tri\Delta,\x}}
{\ruleone{\ruledot{\Delta\vdash\z}}
{\Delta,\x\vdash\z}}
{\Gamma,\x\vdash \,\,\Up S_\x\circ\z}\qquad\qquad
\ruletwo
{\Gamma,\x\vdash W_\x\tri\Gamma}
{\ruletwo
{\ruledot{\Gamma\vdash S\tri\Delta}}
{\ruledot{\Delta\vdash\z}}
{\Gamma\vdash S\circ\z}}
{\Gamma,\x\vdash W_\x\circ S\circ \z}\\[20pt]
Case $W$.\\[10pt]
\ruletwo
{\Gamma,\x\vdash W_\x\tri\Gamma}
{\ruledot{\Gamma\vdash\z}}
{\Gamma,\x\vdash W_\x\circ\z}\qquad\qquad
\ruleone{\ruledot{\Gamma\vdash\z}}
{\Gamma,\x\vdash\z}\\[20pt]
Case $\alpha$.\\[10pt]
\ruleone{\ruledot{\Gamma,\x\vdash A}}
{\Gamma\vdash\lambda\x.A}\qquad\qquad
\ruleone
{\ruletwo
{\Gamma,\y\vdash\{\y\x\}\tri\Gamma,\x}
{\ruledot{\Gamma,\x\vdash A}}
{\Gamma,\y\vdash\{\y\x\}\circ A}}
{\Gamma\vdash\lambda\y.\{\y\x\}\circ A}\\
\end{proof}

\begin{corollary}Reducts of well-formed terms are well-formed.
\end{corollary}

\begin{theorem}\label{decreaseFV}If $A$ is a well-formed term and $A\to B$, then\\ $FV(A)\geqslant FV(B)$.
\end{theorem}
\begin{proof} If $A$ is a well-formed term then $FV(A)\vdash A$ is derivable by Corollary~\ref{FreV}. If also $A\to B$ then $FV(A)\vdash B$ is derivable by the previous theorem and $FV(A)\geqslant FV(B)$ by Corollary~\ref{FreV}.

\end{proof}
\newpage 
\section{The calculus $\lambda\upsilon'$}

\begin{figure}
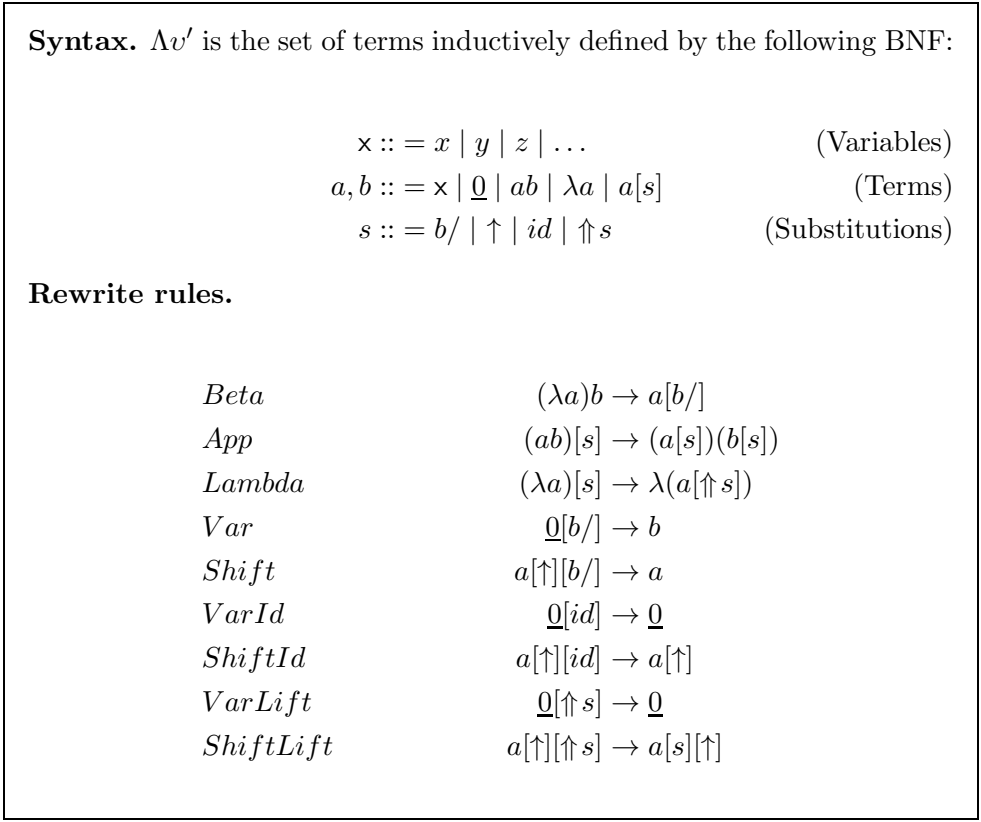

\begin{framed}
\noindent
\textbf{Syntax.} $\Lambda\upsilon'$ is the set of terms inductively defined by the following BNF:

\begin{align*}
& &\x::&= x\mid y\mid z\mid\ldots \tag{Variables}\\
& &a,b::&= \x\mid \un \mid ab \mid \lambda a  \mid  a[s]  \tag{Terms}\\
& &s::&=   b/ \mid \,\,\uparrow\,\, \mid id \mid \,\,\Up s  \tag{Substitutions}
\end{align*}

\textbf{Rewrite rules.}

\begin{align*}
&Beta &  (\lambda a)b &\to a[b/]  \\
&App &  (ab)[s] &\to (a[s])(b[s]) \\
&Lambda & (\lambda a)[s] &\to\lambda (a[\Up s]) \\
&Var &  \un[b/] &\to b \\
&Shift & a[\uparrow][ b/] &\to a \\
&VarId &  \un[id] &\to \un \\
&ShiftId &  a[\uparrow][id] &\to a[\uparrow] \\
&VarLift & \un[\Up s]  &\to \un \\
&ShiftLift & a[\uparrow][\Up s] &\to a[s][\uparrow]
\end{align*}

\end{framed}
\caption{The calculus $\lambda\upsilon'$}\label{upsilon}
\end{figure}

To prove confluence of $\lambda\alpha$, we consider the following calculus $\lambda\upsilon'$.
\begin{definition}The calculus $\lambda\upsilon'$ is shown on Figure~\ref{upsilon}. This calculus contains both named variables  and De Bruijn indices. There are no binders for named variables, they are free in all terms. By $\upsilon'$ we denote $\lambda\upsilon'$ without $Beta$.
\end{definition}

\begin{figure}
\begin{framed}
\noindent
\begin{align*}
\|\x\|_1&= 2 & \|\x\|_2&= 2\\
\|\un\|_1&= 2 & \|\un\|_2&= 2\\
\|ab\|_1&= \|a\|_1+\|b\|_1+1 & \|ab\|_2&= \|a\|_2+\|b\|_2+1\\
\|\lambda a\|_1&= \|a\|_1+1 & \|\lambda a\|_2&= \|a\|_2+1\\
\|a[s]\|_1&= \|a\|_1\cdot \|s\|_1 & \|a[s]\|_2&= \|a\|_2\cdot \|s\|_2\\
\|id\|_1&= 2 & \|id\|_2&= 2\\
\|b/\|_1&= \|b\|_1 & \|b/\|_2&= \|b\|_2\\
\|\uparrow\|_1&= 2 & \|\uparrow\|_2&= 2\\
\|\!\Up s\|_1&= \|s\|_1 & \|\!\Up s\|&= 2\cdot\|s\|_2\\
\end{align*}
\end{framed}
\caption{Interpretations for proving the termination of $\upsilon'$}\label{we}
\end{figure}

\begin{figure}
\begin{framed}
\noindent
\begin{align*}
&R1 && 0\vdash\x \\[5pt]
&R2&& n+1\vdash \un &\\[5pt]
&R3&& \ruletwo{n\vdash a}{n\vdash b}{n\vdash
ab} &\\[10pt]
&R4&& \ruleone{n+1\vdash a}{n\vdash\lambda a} &\\[10pt]
&R5&& \ruletwo{n\vdash s\tri m}{m\vdash
a}{n\vdash
 a[s]} &\\[10pt]
 &R6&& \ruleone{n\vdash
b}{n\vdash
b/\tri n+1} &\\[10pt]
&R7&& n+1\vdash\,\, \uparrow\tri n &\\[5pt]
&R8&& n+1\vdash id \tri n+1 &\\[5pt]
&R9&& \ruleone{n\vdash s\tri m}{n+1\vdash\,\,
\Up s\tri m+1} &\\[10pt]
\end{align*}
\end{framed}
\caption{Inference rules}\label{inf}
\end{figure}

\begin{proposition}The calculus $\upsilon'$ is terminating.
\end{proposition}
\begin{proof}The termination of $\upsilon'$ is proved by a simple lexicographic ordering on two weights $\| \|_1$ and $\| \|_2$ defined on any terms or substitutions (see Figure~\ref{we}). $\| \|_1$ is strictly decreasing on all the rules but $ShiftLift$, on which it is decreasing. $\| \|_2$ is strictly decreasing on $ShiftLift$.\\
\end{proof}

The calculus $\lambda\upsilon'$ is not locally confluent because of the presence of named variables. Now we define  sets of well-formed terms and substitutions to prove confluence on these sets.

\begin{definition}A \emph{judgement} is an expression of the form $n\vdash a$ or of the form  $n\vdash s\tri m$ $(n,m\in \mathbb{N})$. Inference rules for judgements are shown in Figure~\ref{inf}. A term $a$ is \emph{well-formed} iff $n\vdash a$ is derivable for some $n$.
\end{definition}

\begin{lemma}Generation lemma.\\
Each derivation of $n\vdash\x$ is an application of the rule~$R1$, where $n$ is $0$.\\
Each derivation of $n\vdash\un$ is an application of the rule~$R2$, where $n$ is $m+1$ for some $m$.\\
Each derivation of $n\vdash ab$ is an application of the rule~$R3$ to some derivations of $n\vdash a$ and $n\vdash b$.\\
Each derivation of $n\vdash\lambda a$ is an application of the rule~$R4$ to some derivation of $n+1\vdash a$.\\
Each derivation of $n\vdash a[s]$ is an application of the rule~$R5$ to some derivations of $n\vdash s\tri m$ and $m\vdash a$ for some $m$.\\
Each derivation of $n\vdash b/\tri m$ is an application of the rule~$R6$ to some derivation of $n\vdash b$, where $m$ is $n+1$.\\
Each derivation of $n\vdash \,\,\uparrow\tri m$ is an application of the rule~$R7$, where $n$ is $m+1$.\\
Each derivation of $n\vdash id\tri m$ is an application of the rule~$R8$, where $n$ is $k+1$ and $m$ is $k+1$ for some $k$.\\
Each derivation of $n\vdash\,\,\Up s\tri m$ is an application of the rule~$R9$ to some derivation of $k\vdash s\tri l$, where $n$ is $k+1$ and $m$ is $l+1$.
\end{lemma}

\begin{example}$\lambda x$ is not a well-formed term, but $\lambda (x[\uparrow])$ is well-formed.\\ $\lambda\lambda x$ is not a well-formed term, but $\lambda\lambda (x[\uparrow][\uparrow])$ is well-formed.\\ $x[\un/]$ is not a well-formed term, but $x[\uparrow][\un/]$ is well-formed.\\
$x[b/]$ is not a well-formed term, but $x[\uparrow][b/]$ may be well-formed.\\
$x[id]$ is not a well-formed term (see the rule~$R8$), but $x[\uparrow][id]$ is well-formed.\\
$x[\Up s]$ is not a well-formed term, but $x[\uparrow][\Up s]$ may be well-formed.
\end{example}

\begin{corollary}Subterms of well-formed terms are well-formed.
\end{corollary}

\begin{proposition}``Subject reduction''.\\
If $n\vdash a$ and $a\to b$, then $n\vdash b$.
\end{proposition}
\begin{proof}$ $\\[10pt]
Case $Beta$.\\[10pt]
\ruletwo{\ruleone{\ruledot{n+1\vdash a}}{n\vdash\lambda a}}{\ruledot{n\vdash b}}{n\vdash(\lambda a)b}\qquad\qquad\ruletwo{\ruleone{\ruledot{n\vdash b}}{n\vdash b/\tri n+1}}{\ruledot{n+1\vdash a}}{n\vdash a[b/]}\\[20pt]
Case $App$.\\[10pt]
\ruletwo{\ruledot{n\vdash s\tri m}}{\ruletwo{\ruledot{m\vdash a}}{\ruledot{m\vdash b}}{m\vdash ab}}{n\vdash (ab)[s]}\qquad\quad
\ruletwo
{\ruletwo{\ruledot{n\vdash s\tri m}}{\ruledot{m\vdash a}}{n\vdash a[s]}}
{\ruletwo{\ruledot{n\vdash s\tri m}}{\ruledot{m\vdash b}}{n\vdash b[s]}}
{n\vdash(a[s])(b[s])}
\newpage
\noindent
Case $Lambda$.\\[10pt]
   \ruletwo{\ruledot{n\vdash s\tri m}}
  {\ruleone{\ruledot{m+1\vdash a}}{m\vdash\lambda a}}
  {n\vdash (\lambda a)[s]}\qquad\qquad
    \ruleone{\ruletwo{\ruleone{\ruledot{n\vdash s\tri m}}{n+1\vdash\,\,\Up s\tri m+1}}
  {\ruledot{m+1\vdash a}}
  {n+1\vdash a[\Up s]}}
  {n\vdash\lambda a[\Up s]}\\[20pt]
Case $Var$.\\[10pt]
  \ruletwo
 {\ruleone{\ruledot{n\vdash b}}{n\vdash b/\tri n+1}}
 {n+1\vdash \un}
 {n\vdash\un[b/]}\qquad\qquad
  \ruledot{n\vdash b}\\[20pt]
Case $Shift$.\\[10pt]
\ruletwo
{\ruleone{\ruledot{n\vdash b}}
{n\vdash b/\tri n+1}}
{\ruletwo
{n+1\vdash\,\, \uparrow\tri n}
{\ruledot{n\vdash a}}
{n+1\vdash a[\uparrow]}}
{n\vdash a[\uparrow][b/]}\qquad\qquad
\ruledot{n\vdash a}\\[20pt]
Case $VarId$.\\[10pt]
\ruletwo
{n+1\vdash id\tri n+1}
{n+1\vdash\un}
{n+1\vdash \un[id]}\qquad\qquad
$n+1\vdash\un$\\[20pt]
Case $ShiftId$.\\[10pt]
\ruletwo
{n+1\vdash id\tri n+1}
{\ruletwo
{n+1\vdash\,\,\uparrow\tri n}
{\ruledot{n\vdash a}}
{n+1\vdash a[\uparrow]}}
{n+1\vdash a[\uparrow][id]}\qquad\quad
\ruletwo
{n+1\vdash \,\,\uparrow\tri n}
{\ruledot{n\vdash a}}
{n+1\vdash a[\uparrow]}
\newpage
\noindent
Case $VarLift$.\\[10pt]
\ruletwo
{\ruleone{\ruledot{n\vdash s\tri m}}
{n+1\vdash\,\,\Up s\tri m+1}}
{m+1\vdash\un}
{n+1\vdash\un[\Up s]}\qquad\qquad
$n+1\vdash\un$\\[20pt]
Case $ShiftLift$.\\[10pt]
\ruletwo
{\ruleone{\ruledot{n\vdash s\tri m}}
{n+1\vdash\,\,\Up s\tri m+1}}
{\ruletwo
{m+1\vdash\,\,\uparrow\tri m}
{\ruledot{m\vdash a}}
{m+1\vdash a[\uparrow]}}
{n+1\vdash a[\uparrow][\Up s]}\\[20pt]
\ruletwo
{n+1\vdash \,\,\uparrow\tri n}
{\ruletwo
{\ruledot{n\vdash s\tri m}}
{\ruledot{m\vdash a}}
{n\vdash a[s]}}
{n+1\vdash a[s][\uparrow]}

\end{proof}

\begin{corollary}Reducts of well-formed terms are well-formed.
\end{corollary}

\begin{lemma}\label{lemma9}If a well-formed term ``$a$'' is an $\upsilon'$-normal form, then ``$a$'' does not contain substitutions of the forms $b/$, $id$, and $\Up s$.
\end{lemma}

\begin{proof}Induction over the structure of $a$.\\
Suppose, $a$ contains a subterm $a'[b/]$, or $a'[id]$, or $a'[\Up s]$. By induction hypothesis, $a'$ does not contain $[\,/\,]$,\,\,$id$, and $\Up$\,\,, hence $a'$ has the form $c_1c_2$, or $\lambda c$, or $c[\uparrow]$, or $\un$ (by Generation lemma, $a'$ can not be $\x$).
In each case we can apply some rewrite rule, hence $a$ can not be an $\upsilon'$-normal form.\\
\end{proof}

The following five lemmas have similar proofs, I prove the last lemma for example.

\begin{lemma}\label{gor1}
If $a[\Uparrow\!(\uparrow)][\Uparrow\!(b/)]$ is well-formed, there is a common\\ $\upsilon'$-reduct of $a[\Uparrow\!(\uparrow)][\Uparrow\!(b/)]$ and $a$.
\end{lemma}

\begin{lemma}\label{gor2}If $a[\Uparrow\!(\uparrow)][\Uparrow\!(id)]$ is well-formed, there is a common \\ $\upsilon'$-reduct of $a[\Uparrow\!(\uparrow)][\Uparrow\!(id)]$ and $a[\Uparrow\!(\uparrow)]$.
\end{lemma}

\begin{lemma}\label{gor3}If $a[\Uparrow\!\!(\uparrow)][\Uparrow\Uparrow\!\! s]$ and $a[\Uparrow\!\! s][\Uparrow\!\!(\uparrow)]$ are well-formed, there is a common $\upsilon'$-reduct of these terms.
\end{lemma}

\begin{lemma}\label{lemma10}If $a[b/][s]$ and $a[\Up s][b[s]/]$ are well-formed, there is a common $\upsilon'$-reduct of these terms.
\end{lemma}

\begin{lemma}\label{lemmaid}If $a[id]$ is well-formed, there is  a common $\upsilon'$-reduct of $a[id]$ and $a$.
\end{lemma}
\begin{proof}We prove the following stronger result: if $a[\Up\!^n id]\,\, (n\geqslant 0)$ is well-formed, there is  a common $\upsilon'$-reduct of $a[\Up\!^n id]$ and $a$. The proof is by induction over the structure of $a$. By Lemma~\ref{lemma9}, we can assume that $a$ does not contain $[\,/\,]$, $id$, and $\Up$\,.\\[5pt]
Case 1. $a$ has the form $a_1a_2$.\\
$(a_1a_2)[\Up^n id]\overset{App}{\to}(a_1[\Up^n id])(a_2[\Up^n id])$\\
Then we use the induction hypothesis.\\[5pt]
Case 2. $a$ has the form $\lambda a'$.\\
$(\lambda a')[\Up^n id]\overset{Lambda}{\to}\lambda (a'[\Up^{n+1} id])$\\
Then we use the induction hypothesis.\\[5pt]
Case 3. $a$ is $\un$.\\
If $n=0$, then\\
$\un[id]\overset{VarId}{\to}\un$\\
If $n=m+1$, then\\
$\un[\Up^{m+1}id]\overset{VarLift}{\to}\un$\\[5pt]
Case 4. $a$ has the form $a'[\uparrow]$.\\
If $n=0$, then\\
$ a'[\uparrow][id]\overset{ShiftId}{\to}a'[\uparrow]$\\
If $n=m+1$, then\\
$a'[\uparrow][\Up^{m+1} id]\overset{ShiftLift}{\to}a'[\Up^m id][\uparrow]$\\
Then we use the induction hypothesis.\\[5pt]
Note that $a$ can not be $\x$ by Generation lemma.\\
\end{proof}
\begin{theorem}The rewriting system $\upsilon'$ is locally confluent (hence, confluent) on the set of well-formed
terms.
\end{theorem}
\begin{proof}Straightforward checking, using Lemma~\ref{gor1}, Lemma~\ref{gor2}, and Lemma~\ref{gor3} in the following cases
$$\begin{tikzcd}
 (\lambda a)[\uparrow][b/]  \arrow{rr}{Shift} \arrow{dd}[swap]{Lambda} && \lambda a \\
  & & \\
   (\lambda(a[\Up(\uparrow)]))[b/]  &&
 \end{tikzcd}$$

$$\begin{tikzcd}
 (\lambda a)[\uparrow][id]  \arrow{rr}{ShiftId} \arrow{dd}[swap]{Lambda} && (\lambda a)[\uparrow] \\
  & & \\
   (\lambda(a[\Up(\uparrow)]))[id]  &&
 \end{tikzcd}$$

$$\begin{tikzcd}
 (\lambda a)[\uparrow][\Up s]  \arrow{rr}{ShiftLift} \arrow{dd}[swap]{Lambda} && (\lambda a)[s][\uparrow] \\
  & & \\
   (\lambda(a[\Up(\uparrow)]))[\Up s]  &&
 \end{tikzcd}$$

\end{proof}
\begin{lemma}\label{lemma11}Let $R$ and $S$ be two relations defined on the same set $X$, $R$ is confluent and strongly normalizing, and $S$ verifying the diamond property:
$$\begin{tikzcd}
 f  \arrow{rr}{S} \arrow{dd}[swap]{S} && g \arrow[dashed]{dd}{S}\\
  & & \\
   h \arrow[dashed]{rr}{S} && k
 \end{tikzcd}$$
 Suppose moreover that the following diagram holds:
 $$\begin{tikzcd}
 f  \arrow{rr}{S} \arrow{dd}[swap]{R} && g \arrow[dashed]{dd}{R^*}\\
  & & \\
   h \arrow[dashed]{rr}{R^*SR^*} && k
 \end{tikzcd}$$
 Here $R^*$ is the reflexive and transitive closure of $R$.
 Then the relation $R^*SR^*$ is confluent.
 \end{lemma}
 \begin{proof} See~\cite{Hardin} (Lemma 4.5).\\
 \end{proof}
 We shall apply the lemma with the following data. We take the set of well-formed terms as $X$, $\upsilon'$ as $R$, and $Beta\!\parallel$ as $S$, where $Beta\!\parallel$ is the obvious parallelization of $Beta$ defined by:\\[10pt]
 $\begin{array}{ccc}
  a\to a && s\to s\\[15pt]
\ruletwo{a_1\to a_2}{b_1\to b_2}{(\lambda a_1)b_1\to a_2[b_2/]} &&  \ruleone{ a_1\to a_2}{ \lambda
a_1\to\lambda a_2}\\[20pt]
 \ruletwo{ a_1\to a_2}{ b_1\to b_2}{ a_1b_1\to  a_2b_2} && \ruletwo{  a_1\to a_2}{  s_1\to s_2}{a_1[s_1]\to a_2[s_2]}\\[20pt]
 \ruleone{ b_1\to b_2}{ b_1/\to b_2/} && \ruleone{ s_1\to s_2}{ \Up s_1\to\,\,\Up s_2}
 \end{array}$\\[20pt]

  \begin{proposition}$\upsilon'$ and $Beta\!\parallel$ satisfy the conditions of Lemma~\ref{lemma11}.
 \end{proposition}
 \begin{proof}The strong confluence of $Beta\!\parallel$ is obvious since $Beta$ by itself is a left linear system with no critical pairs.
 Now we check the second diagram.\\
 Case $App$. $f\equiv (ab)[s]\overset{\upsilon'}\to (a[s])(b[s])\equiv h$. Then there are two cases:\\
 1. $f\equiv (ab)[s]\overset{Beta\parallel}\to (a'b')[s']\equiv g$ with $a\overset{Beta\parallel}\to a'$, $b\overset{Beta\parallel}\to b'$, and $s\overset{Beta\parallel}\to s'$. Then by definition of $Beta\!\parallel$ we have $(a[s])(b[s])\overset{Beta\parallel}\to(a'[s'])(b'[s'])\equiv k$. But also $g\overset{\upsilon'}\to k$.\\
 2. $f\equiv ((\lambda a)b)[s]\overset{Beta\parallel}\to a'[b'][s']\equiv g$ with $a\overset{Beta\parallel}\to a'$, $b\overset{Beta\parallel}\to b'$, and $s\overset{Beta\parallel}\to s'$. Then $h \equiv ((\lambda a)[s])(b[s])$. We must then take $h \overset{\upsilon'}\to (\lambda (a[\Up s]))(b[s])\equiv h_1$. Then $h_1\overset{Beta\parallel}\to a'[\Up s'][b'[s'/]]\equiv h_2$. Using Lemma~\ref{lemma10}, we check that $h_2\overset{\upsilon'^*}\to k$ and $g\overset{\upsilon'^*}\to k$ for some $k$. This subcase is the only interesting one.\\
 The cases of all other rewrite rules are simple and similar to subcase 1.\\
 \end{proof}
 \begin{theorem}The rewriting system $\lambda\upsilon'$ is confluent on the set of well-formed terms.
 \end{theorem}
 \begin{proof}$\lambda\upsilon'\subseteq R^*SR^*\subseteq\lambda\upsilon'^*$.\\
 \end{proof}

\newpage 
\section{$\alpha$-conversion and confluence}
\begin{figure}
\begin{framed}
\noindent

\begin{align*}
& &&\| G\vdash \x\|=  \x & (\x\in G)\\[15pt]
&&&\|\Gamma,\x\vdash \x\| =  \un \\[15pt]
&&&\ruleone{\|\Gamma\vdash
\x\|=  a}{\|\Gamma,\y\vdash \x\|=   a[\uparrow]} & (\x\neq \y)\\[20pt]
&&&\ruletwo{\|\Gamma\vdash A\|=   a}{\|\Gamma\vdash B\|=  b}{\|\Gamma\vdash
AB\|=   ab}\\[15pt]
&&&\ruleone{\|\Gamma,\x\vdash A\|=   a}{\|\Gamma\vdash\lambda \x.A\|=    \lambda a}\\[15pt]
&&&\ruletwo{\|\Gamma\vdash S \tri\Delta\|=  s }{\|\Delta\vdash A\|=   a}{\|\Gamma\vdash
 S \circ A\|=    a[s]}\\[15pt]
&&&\ruleone{\|\Gamma\vdash B\|=   b}{\|\Gamma\vdash [B/\x]\tri\Gamma,\x\|=  b/ }\\[15pt]
&&&\|\Gamma,\x\vdash W_\x \tri \Gamma\|=    \,\,\uparrow \\[15pt]
&&&\|\Gamma,\y\vdash \{\y\x\} \tri \Gamma,\x\|=   id \\[15pt]
&&&\ruleone{\|\Gamma\vdash S\tri\Delta\|=  s}{\|\Gamma,\x\vdash \,\,\Up S_\x  \tri\Delta,\x\|=
\,\,\Up s}
\end{align*}
\end{framed}
\caption{Correspondence}\label{correspondence}
\end{figure}

Recall that each derivable judgement has a unique derivation (Proposition~\ref{unique}).
\begin{definition} We associate with every derivable judgement $\Gamma\vdash A$  some
$\lambda\upsilon'$-term $\|\Gamma\vdash A\|$ as it  is shown on Figure~\ref{correspondence}.\\
 We associate with every derivable judgement $\Gamma\vdash S\tri\Delta$  some
 $\lambda\upsilon'$-substitution  $\|\Gamma\vdash S\tri\Delta\|$ as it is shown on Figure~\ref{correspondence}.
\end{definition}

\begin{example}
$$\ruleone{\ruletwo{\|\{x\},x\vdash W_x\tri\{x\}\|=\,\,\uparrow}{\|\{x\}\vdash x\|= x}{\|\{x\},x\vdash
W_x\circ x\|= x[\uparrow]}}{\|\{x\}\vdash\lambda x.W_x\circ x\|=\lambda (x[\uparrow])}$$
\end{example}

\begin{example}
$$\ruleone
{\ruletwo
{\|\emptyset,y\vdash\{yx\}\tri\emptyset,x\|= id}
{\|\emptyset,x\vdash x\|=\un}
{\|\emptyset,y\vdash \{yx\}\circ x\|=\un[id]}}
{\|\emptyset\vdash\lambda y.\{yx\}\circ x\|=\lambda(\un[id])}
$$
\end{example}

\begin{proposition} If $\|\Gamma\vdash A\|= a$, then $a$ is well-formed.\\ If $\|\Gamma\vdash
S\tri\Delta\|= s$, then $s$ is well-formed.
\end{proposition}
\begin{proof} Easy induction  shows that if $\|\Gamma\vdash A\|= a$,
 then $n\vdash a$ is derivable, where $n$ is the length of local part of $\Gamma$. Similarly, if
 $\|\Gamma\vdash S\tri\Delta\|= s$,
 than $n\vdash s\tri m$ is derivable, where $n$ is the length of local part of $\Gamma$ and $m$ is
 the length of local part of $\Delta$.\\
 \end{proof}

 \begin{example}
 $$\ruleone{\ruletwo{\{x\},x\vdash W_x\tri\{x\}}{\{x\}\vdash x}{\{x\},x\vdash
 W_x\circ x}}{\{x\}\vdash\lambda x.W_x\circ x}\qquad
 \ruleone{\ruletwo{1\vdash \,\,\uparrow\tri\, 0}{0\vdash x}{1\vdash
 x[\uparrow]}}{0\vdash\lambda (x[\uparrow])}$$
 \end{example}

\begin{definition}We write $A\equiv_{\Gamma}B$ iff
 $\|\Gamma\vdash A\|$ is the same $\lambda\upsilon'$-term as
     $\|\Gamma\vdash B\|$.
      \end{definition}

   Note that if  $\|\Gamma\vdash A\|=  a$, then $\Gamma\vdash A$ is derivable. Hence
   $A\equiv_{\Gamma}B$ implies $A$ and $B$ are well-formed.

      \begin{example}If $\Gamma=\{x\},y$, then $W_y\circ x\equiv_{\Gamma}x$. Both terms correspond
      to $ x[\uparrow]$.
   \end{example}

   \begin{example}$\lambda x.W_x\circ x\equiv_{\{x\}}\lambda y.W_y\circ x\equiv_{\{x\}}\lambda y.
   x$\\
   All these terms correspond to $\lambda (x[\uparrow])$. But $\lambda x.W_x\circ
   x\not\equiv_{\{x\}}\lambda x.x$, because $\lambda x.x$ corresponds to $\lambda\un$.
   \end{example}

 \begin{definition} We write $A\equiv_{\alpha}B$ iff $FV(A)=FV(B)$ and
 $A\equiv_{\Gamma} B$, where $\Gamma=FV(A)=FV(B)$.
 \end{definition}

 \begin{example} $\lambda y.xy\equiv_{\alpha}\lambda z.xz$
 \end{example}

 \begin{example}$\lambda x.W_x\circ x\equiv_{\alpha}\lambda y.W_y\circ x\equiv_{\alpha}\lambda y.x\not\equiv_{\alpha}\lambda x.
 x$
   \end{example}

   Now we shall prove confluence of $\lambda\alpha$ in the following form:

$$
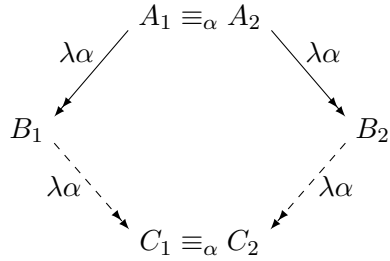
\begin{tikzpicture}
\node (K1) at (0,0) {$A_1\equiv_{\alpha} A_2$};
\node (K2) at (-2.3,-1.5) {$B_1$};
\node (K3) at (2.3,-1.5) {$B_2$};
\node (K4) at (0,-3) {$C_1\equiv_{\alpha} C_2$};
\begin{scope}[>=latex]
\path[->>] ([yshift= -5pt] K1.west) edge node[left=2pt,near start]{$\lambda\alpha$} ([yshift= 5pt]
K2.east) ;
\path[->>] ([yshift= -5pt] K1.east) edge node[right=2pt,near start]{$\lambda\alpha$}([yshift= 5pt]
K3.west) ;
\path[->>][dashed] ([yshift= -5pt] K2.east) edge node[left]{$\lambda\alpha$}([yshift= 5pt] K4.west)
;
\path[->>][dashed] ([yshift= -5pt] K3.west) edge node[right]{$\lambda\alpha$}([yshift= 5pt]
K4.east) ;
\end{scope}
\end{tikzpicture}$$

\begin{lemma}If the following conditions hold
\begin{itemize}
\item $\|\Gamma\vdash A\|=  a$
\item $A\overset{W}{\to}B$
\end{itemize}
then $\|\Gamma\vdash B\|=  a$
\end{lemma}

\begin{proof} If $A$ contains a $W$-redex $W_\x\circ\z\quad(\x\neq\z$), then the unique derivation
of $\Gamma\vdash A$ contains a sub-derivation of the form\\[10pt]
$$\ruletwo{\Delta,\x\vdash W_\x\tri\Delta}{\ruledot{\Delta\vdash\z}}{\Delta,\x\vdash
W_\x\circ\z}$$\\[10pt]
and the unique derivation of $\Gamma\vdash B$ contains instead of it the sub-derivation\\[10pt]
$$\ruleone{\ruledot{\Delta\vdash\z}}{\Delta,\x\vdash\z}$$\\[10pt]
Suppose $\|\Delta\vdash\z\|=  a'$, then $\|\Delta,\x\vdash W_\x\circ\z\|=   a'[\uparrow]$ and
$\|\Delta,\x\vdash\z\|=   a'[\uparrow]$\\
\end{proof}

\begin{definition}We denote by $\sigma$ the calculus $\lambda\alpha$ without the rules $Beta$, $W$, and
$\alpha$.
\end{definition}

\begin{lemma}If the following conditions hold
\begin{itemize}
  \item $\|\Gamma\vdash A\|=  a$
   \item $A\overset{\sigma\cup\{Beta\}}{\to}B$
  \item $\|\Gamma\vdash B\|=  b$
\end{itemize}
then $a\overset{\lambda\upsilon'}{\to}b$.
\end{lemma}

\begin{proof}The rules of $\sigma\cup\{Beta\}$  correspond to the rules of
$\lambda\upsilon'$. For example, consider the rule $Shift'$. Suppose $A\overset {Shift'}{\to}B$.
Then $A$ contains a redex $[C/\x]\circ\z\,\,\,(\x\neq\z)$. The derivation of $\Gamma\vdash A$ must
contain a sub-derivation of the form\\[10pt]
$$\ruletwo
{\ruleone{\ruledot{\Delta\vdash C}}{\Delta\vdash[C/\x]\tri\Delta,\x}}
{\ruleone{\ruledot{\Delta\vdash\z}}{\Delta,\x\vdash\z}}
{\Delta\vdash[C/\x]\circ\z}$$\\[10pt]
The derivation of $\Gamma\vdash B$ contains instead of it the sub-derivation\\[10pt]
$${\ruledot{\Delta\vdash\z}}$$\\[10pt]
Suppose $\|\Delta\vdash C\|=  c$ and $\|\Delta\vdash\z\|=  a'$. Then  $\|\Delta\vdash[C/\x]\circ\z\|=
a'[\uparrow][c/]$
  $$\ruletwo
 {\ruleone{\ruledot{\|\Delta\vdash C\|=  c}}
 {\|\Delta\vdash[C/\x]\tri\Delta,\x\|= [c/]}}
 {\ruleone{\ruledot{\|\Delta\vdash\z\|=  a'}}
 {\|\Delta,\x\vdash\z\|=  a'[\uparrow]}}
 {\|\Delta\vdash[C/\x]\circ\z\|=  a'[\uparrow][c/]}$$

We obtain\\[5pt]
$a'[\uparrow][c/]\overset{Shift}{\to}a'$, hence $a\overset{Shift}{\to}b$\\
\end{proof}

\begin{corollary}\label{tuda}If the following conditions hold
\begin{itemize}
  \item $\|\Gamma\vdash A\|=  a$
   \item $A\overset{\sigma\cup\{Beta,W\}}{-\!\!\!\twoheadrightarrow}B$
  \item $\|\Gamma\vdash B\|=  b$
\end{itemize}
then $a\overset{\lambda\upsilon'}{\twoheadrightarrow}b$.
\end{corollary}

\begin{lemma}\label{suda} If the following conditions hold
\begin{itemize}
  \item $\|\Gamma\vdash B\|=  b$
  \item $b\overset{\lambda\upsilon'}{\twoheadrightarrow}c$
\end{itemize}
then there exists a term $C$ such that
\begin{itemize}
\item $\|\Gamma\vdash C\|=  c$
  \item $B\overset{\sigma\cup\{Beta\}}{-\!\!\!\twoheadrightarrow}C$
  \end{itemize}
\end{lemma}

\begin{proof}The rules of $\lambda\upsilon'$ corresponds to the rules of $\sigma\cup\{Beta\}$.
\end{proof}

   \begin{theorem}\label{gexagon} Suppose
   \begin{itemize}
     \item $A_1\equiv_{\Gamma}A_2$
     \item $A_1\overset{\lambda\alpha}{\twoheadrightarrow} B_1$
     \item $A_2\overset{\lambda\alpha}{\twoheadrightarrow} B_2$
   \end{itemize}
      then there are  terms $C_1$ and $C_2$ such that
      \begin{itemize}
        \item $B_1\overset{\lambda\alpha}{\twoheadrightarrow} C_1$
        \item $B_2\overset{\lambda\alpha}{\twoheadrightarrow} C_2$
        \item $C_1\equiv_{\Gamma} C_2$
      \end{itemize}

$$
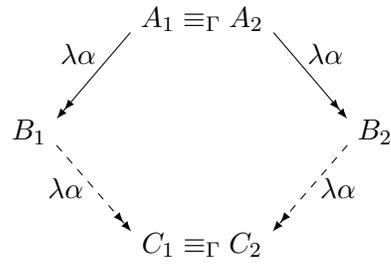
\begin{tikzpicture}
\node (K1) at (0,0) {$A_1\equiv_{\Gamma} A_2$};
\node (K2) at (-2.3,-1.5) {$B_1$};
\node (K3) at (2.3,-1.5) {$B_2$};
\node (K4) at (0,-3) {$C_1\equiv_{\Gamma} C_2$};
\begin{scope}[>=latex]
\path[->>] ([yshift= -5pt] K1.west) edge node[left=2pt,near start]{$\lambda\alpha$} ([yshift= 5pt]
K2.east) ;
\path[->>] ([yshift= -5pt] K1.east) edge node[right=2pt,near start]{$\lambda\alpha$}([yshift= 5pt]
K3.west) ;
\path[->>][dashed] ([yshift= -5pt] K2.east) edge node[left]{$\lambda\alpha$}([yshift= 5pt] K4.west)
;
\path[->>][dashed] ([yshift= -5pt] K3.west) edge node[right]{$\lambda\alpha$}([yshift= 5pt]
K4.east) ;
\end{scope}
\end{tikzpicture}$$
 \end{theorem}
 \begin{proof}We prove the following stronger result

 $$
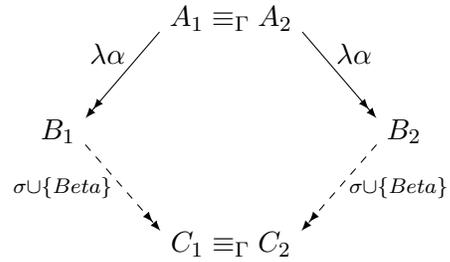
\begin{tikzpicture}
\node (K1) at (0,0) {$A_1\equiv_{\Gamma} A_2$};
\node (K2) at (-2.3,-1.5) {$B_1$};
\node (K3) at (2.3,-1.5) {$B_2$};
\node (K4) at (0,-3) {$C_1\equiv_{\Gamma} C_2$};
\begin{scope}[>=latex]
\path[->>] ([yshift= -5pt] K1.west) edge node[left=2pt,near start]{$\lambda\alpha$} ([yshift= 5pt]
K2.east) ;
\path[->>] ([yshift= -5pt] K1.east) edge node[right=2pt,near start]{$\lambda\alpha$}([yshift= 5pt]
K3.west) ;
\path[->>][dashed] ([yshift= -5pt] K2.east) edge
node[left]{$\scriptstyle{\sigma\cup\{Beta\}}$}([yshift= 5pt] K4.west) ;
\path[->>][dashed] ([yshift= -5pt] K3.west) edge
node[right]{$\scriptstyle{\sigma\cup\{Beta\}}$}([yshift= 5pt] K4.east) ;
\end{scope}
\end{tikzpicture}$$
 Suppose

 \begin{itemize}
   \item $\|\Gamma\vdash A_1\|=  a$
   \item $\|\Gamma\vdash A_2\|=  a$
   \item $\|\Gamma\vdash B_1\|=  b_1$
   \item $\|\Gamma\vdash B_2\|=  b_2$
 \end{itemize}
 $ $\\
  Case 1.$ $\\

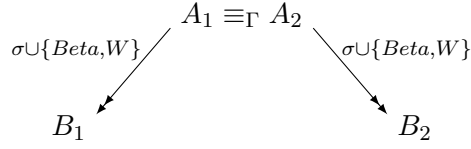
 $$\begin{tikzpicture}
\node (K1) at (0,0) {$A_1\equiv_{\Gamma} A_2$};
\node (K2) at (-2.3,-1.5) {$B_1$};
\node (K3) at (2.3,-1.5) {$B_2$};
\begin{scope}[>=latex]
\path[->>] ([yshift= -5pt] K1.west) edge node[left,near start]{$\scriptstyle{\sigma\cup\{Beta, W\}}$}
([yshift= 5pt] K2.east) ;
\path[->>] ([yshift= -5pt] K1.east) edge node[right,near
start]{$\scriptstyle{\sigma\cup\{Beta, W\}}$}([yshift= 5pt] K3.west) ;
\end{scope}
\end{tikzpicture}$$
Using Corollary~\ref{tuda} and confluence of $\lambda\upsilon'$ on the set of well-formed terms, we
obtain

 $$
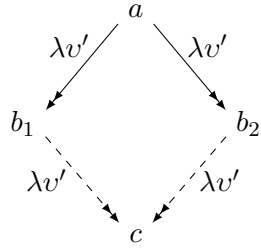
\begin{tikzpicture}
\node (K1) at (0,0) {$a$};
\node (K2) at (-1.5,-1.5) {$b_1$};
\node (K3) at (1.5,-1.5) {$b_2$};
\node (K4) at (0,-3) {$c$};
\begin{scope}[>=latex]
\path[->>] ([yshift= -5pt] K1.west) edge node[left,near start]{$\lambda\upsilon'$} ([yshift= 5pt]
K2.east) ;
\path[->>] ([yshift= -5pt] K1.east) edge node[right=2pt,near start]{$\lambda\upsilon'$}([yshift=
5pt] K3.west) ;
\path[->>][dashed] ([yshift= -5pt] K2.east) edge node[left=2pt ]{$\lambda\upsilon'$}([yshift= 5pt]
K4.west) ;
\path[->>][dashed] ([yshift= -5pt] K3.west) edge node[right ]{$\lambda\upsilon'$}([yshift= 5pt]
K4.east) ;
\end{scope}
\end{tikzpicture}$$
Then we use Lemma~\ref{suda}.\\[5pt]
Case 2.$ $\\

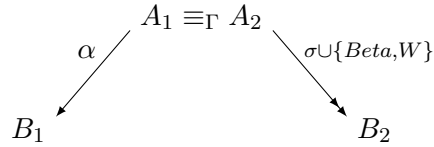
  $$\begin{tikzpicture}
\node (K1) at (0,0) {$A_1\equiv_{\Gamma} A_2$};
\node (K2) at (-2.3,-1.5) {$B_1$};
\node (K3) at (2.3,-1.5) {$B_2$};
\begin{scope}[>=latex]
\path[->] ([yshift= -5pt] K1.west) edge node[left=2pt,near start]{$\alpha$} ([yshift= 5pt] K2.east)
;
\path[->>] ([yshift= -5pt] K1.east) edge node[right=1pt,near
start]{$\scriptstyle{\sigma\cup\{Beta, W\}}$}([yshift= 5pt] K3.west) ;
\end{scope}
\end{tikzpicture}$$
The term $A_1$ contains an $\alpha$-redex of the form $\lambda\x.A'$. The term $B_1$ contains
instead of it a subterm $\lambda\y.\{\y\x\}\circ A'$. The derivation of $\Gamma\vdash A_1$ contains
a sub-derivation of the form $\Delta\vdash\lambda\x.A'$\\ Suppose $\|\Delta\vdash\lambda\x.A'\|=
\lambda a'$\\
Then $\|\Delta\vdash\lambda\y.\{\y\x\}\circ A'\|= \lambda (a'[id])$\\
Using Lemma~\ref{lemmaid}, we obtain

 $$
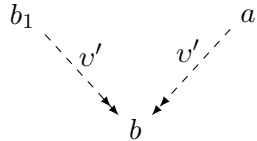
\begin{tikzpicture}
\node (K0) at (-3,0) {$b_1$};
\node (K1) at (0,0) {$a$};
\node (K2) at (-1.5,-1.5) {$b$};
\begin{scope}[>=latex]
\path[->>][dashed] ([yshift= -7pt] K0.east) edge node[right=2pt,near start]{$\upsilon'$}([yshift=
5pt] K2.west) ;
\path[->>][dashed] ([yshift= -5pt] K1.west) edge node[left,near start]{$\upsilon'$} ([yshift= 5pt]
K2.east) ;
\end{scope}
\end{tikzpicture}$$
for some $b$, hence

 
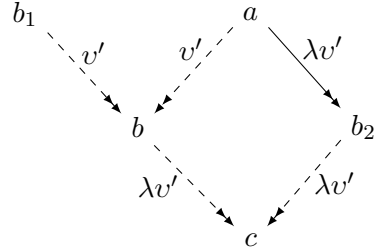
$$\begin{tikzpicture}
\node (K0) at (-3,0) {$b_1$};
\node (K1) at (0,0) {$a$};
\node (K2) at (-1.5,-1.5) {$b$};
\node (K3) at (1.5,-1.5) {$b_2$};
\node (K4) at (0,-3) {$c$};
\begin{scope}[>=latex]
\path[->>][dashed] ([yshift= -7pt] K0.east) edge node[right=2pt,near start]{$\upsilon'$}([yshift=
5pt] K2.west) ;
\path[->>][dashed] ([yshift= -5pt] K1.west) edge node[left,near start]{$\upsilon'$} ([yshift= 5pt]
K2.east) ;
\path[->>] ([yshift= -5pt] K1.east) edge node[right=2pt,near start]{$\lambda\upsilon'$}([yshift=
7pt] K3.west) ;
\path[->>][dashed] ([yshift= -7pt] K2.east) edge node[left=2pt ]{$\lambda\upsilon'$}([yshift= 5pt]
K4.west) ;
\path[->>][dashed] ([yshift= -5pt] K3.west) edge node[right ]{$\lambda\upsilon'$}([yshift= 5pt]
K4.east) ;
\end{scope}
\end{tikzpicture}$$
Then we use Lemma~\ref{suda}.\\[5pt]
  Case 3.$ $\\

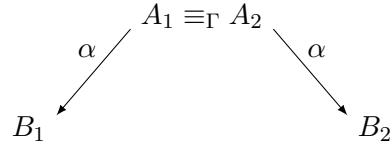
  $$\begin{tikzpicture}
\node (K1) at (0,0) {$A_1\equiv_{\Gamma} A_2$};
\node (K2) at (-2.3,-1.5) {$B_1$};
\node (K3) at (2.3,-1.5) {$B_2$};
\begin{scope}[>=latex]
\path[->] ([yshift= -5pt] K1.west) edge node[left=2pt,near start]{$\alpha$} ([yshift= 5pt] K2.east)
;
\path[->] ([yshift= -5pt] K1.east) edge node[right=2pt,near start]{$\alpha$}([yshift= 5pt] K3.west)
;
\end{scope}
\end{tikzpicture}$$
As in the previous case.

$$
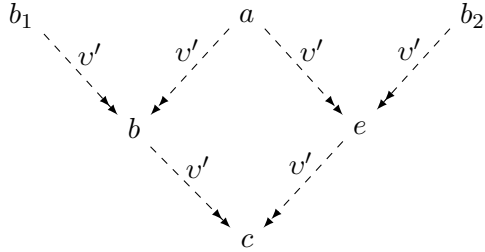
\begin{tikzpicture}
\node (K) at (3,0) {$b_2$};
\node (K0) at (-3,0) {$b_1$};
\node (K1) at (0,0) {$a$};
\node (K2) at (-1.5,-1.5) {$b$};
\node (K3) at (1.5,-1.5) {$e$};
\node (K4) at (0,-3) {$c$};
\begin{scope}[>=latex]
\path[->>][dashed] ([yshift= -7pt] K0.east) edge node[right=2pt,near start]{$\upsilon'$}([yshift=
5pt] K2.west) ;
\path[->>][dashed] ([yshift= -5pt] K.west) edge node[left,near start]{$\upsilon'$} ([yshift= 7pt]
K3.east) ;
\path[->>][dashed] ([yshift= -5pt] K1.west) edge node[left,near start]{$\upsilon'$} ([yshift= 5pt]
K2.east) ;
\path[->>][dashed] ([yshift= -5pt] K1.east) edge node[right=2pt,near start]{$\upsilon'$}([yshift=
5pt] K3.west) ;
\path[->>][dashed] ([yshift= -7pt] K2.east) edge node[right=2pt,near start]{$\upsilon'$}([yshift=
5pt] K4.west) ;
\path[->>][dashed] ([yshift= -5pt] K3.west) edge node[left,near start]{$\upsilon'$}([yshift= 5pt]
K4.east) ;
\end{scope}
\end{tikzpicture}$$
 \end{proof}

\begin{definition} A term $A$ is  \emph{good} iff there is a global context $G$ such that $G\vdash
 A$ is derivable (the local context is empty).
 \end{definition}

 \begin{example}\label{bedsubst}A term of the form $W_{\x}\circ A$ can not be good
 $$
 \ruletwo{\Delta,\x\vdash W_{\x}\tri\Delta}{\ruledot{\Delta\vdash A}}{\Delta,\x\vdash W_{\x}\circ A}
 $$
 \end{example}

 But a term of the form $\lambda \x.W_{\x}\circ A$ can be good. In a good term, each symbol $W$ must
 be ``killed'' by lambda or another binder.

 \begin{example}The term $\lambda xy.W_y\circ W_x\circ z$ is good.
 \end{example}

 \begin{proposition} Each usual lambda-term is a good term.
 \end{proposition}

 \begin{proof} Because $FV(A)\vdash A$ is derivable for each usual $A$.
 \end{proof}

 \begin{proposition}All reducts of good terms are good.
 \end{proposition}
 \begin{proof}By Theorem~\ref{subbred}.
 \end{proof}

 \begin{lemma} If $A$ is a good term then $FV(A)$ is uniquely determined by $\|FV(A)\vdash A\|$.
 \end{lemma}

 \begin{proof} $FV(A)$ is a set of all variables in $\|FV(A)\vdash A\|$.
 \end{proof}

 \begin{example} $\|\{y,z\}\vdash \lambda x.xyz\|=\lambda \un yz$
 \end{example}

 \begin{lemma} If $A$ is a good term and $FV(A)\subseteq G$ then $\|FV(A)\vdash A\|$ is the same $\lambda\upsilon'$-term as $\|G\vdash A\|$.
 \end{lemma}

 \begin{proof} ``The same'' derivation holds. Note that $G$ is a set, not an arbitrary context.
 \end{proof}

\begin{theorem} Suppose
   \begin{itemize}
   \item $A_1$ and $A_2$ are good terms;
     \item $A_1\equiv_{\alpha}A_2$
     \item $A_1\overset{\lambda\alpha}{\twoheadrightarrow} B_1$
     \item $A_2\overset{\lambda\alpha}{\twoheadrightarrow} B_2$
   \end{itemize}
      then there are  terms $C_1$ and $C_2$ such that
      \begin{itemize}
        \item $B_1\overset{\lambda\alpha}{\twoheadrightarrow} C_1$
        \item $B_2\overset{\lambda\alpha}{\twoheadrightarrow} C_2$
        \item $C_1\equiv_{\alpha} C_2$
      \end{itemize}

\end{theorem}

\begin{proof} Let $G=FV(A_1)=FV(A_2)$. By Theorem~\ref{gexagon} we obtain

$$
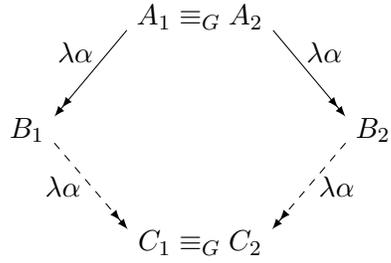
\begin{tikzpicture}
\node (K1) at (0,0) {$A_1\equiv_{G} A_2$};
\node (K2) at (-2.3,-1.5) {$B_1$};
\node (K3) at (2.3,-1.5) {$B_2$};
\node (K4) at (0,-3) {$C_1\equiv_{G} C_2$};
\begin{scope}[>=latex]
\path[->>] ([yshift= -5pt] K1.west) edge node[left=2pt,near start]{$\lambda\alpha$} ([yshift= 5pt]
K2.east) ;
\path[->>] ([yshift= -5pt] K1.east) edge node[right=2pt,near start]{$\lambda\alpha$}([yshift= 5pt]
K3.west) ;
\path[->>][dashed] ([yshift= -5pt] K2.east) edge node[left]{$\lambda\alpha$}([yshift= 5pt] K4.west)
;
\path[->>][dashed] ([yshift= -5pt] K3.west) edge node[right]{$\lambda\alpha$}([yshift= 5pt]
K4.east) ;
\end{scope}
\end{tikzpicture}$$
By the previous two lemmas  $\|FV(C_1)\vdash C_1\|$ is the same $\lambda\upsilon'$-term as \\ $\|FV(C_2)\vdash C_2\|$ and $FV(C_1)=FV(C_2)$, hence $C_1\equiv_{\alpha} C_2$

$$
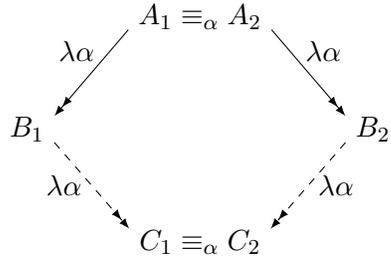
\begin{tikzpicture}
\node (K1) at (0,0) {$A_1\equiv_{\alpha} A_2$};
\node (K2) at (-2.3,-1.5) {$B_1$};
\node (K3) at (2.3,-1.5) {$B_2$};
\node (K4) at (0,-3) {$C_1\equiv_{\alpha} C_2$};
\begin{scope}[>=latex]
\path[->>] ([yshift= -5pt] K1.west) edge node[left=2pt,near start]{$\lambda\alpha$} ([yshift= 5pt]
K2.east) ;
\path[->>] ([yshift= -5pt] K1.east) edge node[right=2pt,near start]{$\lambda\alpha$}([yshift= 5pt]
K3.west) ;
\path[->>][dashed] ([yshift= -5pt] K2.east) edge node[left]{$\lambda\alpha$}([yshift= 5pt] K4.west)
;
\path[->>][dashed] ([yshift= -5pt] K3.west) edge node[right]{$\lambda\alpha$}([yshift= 5pt]
K4.east) ;
\end{scope}
\end{tikzpicture}$$

\end{proof}

\begin{note} Confluence holds for all well-formed terms (not only good) but the proof is more complicated.
\end{note}

\newpage

\section{Normal forms}
\begin{figure}
\begin{framed}
\noindent
\textbf{Syntax.} The set of $\sigma\cup\{W\}$-normal forms is inductively defined by the following BNF:
\begin{align*}
\x,\y,\z::&= x\mid y\mid z\mid\ldots \tag{Variables}\\
\D::&=W_\z\circ\z \mid W_\x\circ\D \tag{Blocks}\\
A,B::&= \x \mid \D \mid AB \mid \lambda\x.A   \tag{Terms}
\end{align*}
\end{framed}
\caption{$\sigma\cup\{W\}$-normal forms}\label{sigmanf}
\end{figure}

\begin{lemma}\label{nfW} A well-formed term $A$ is a $\sigma\cup\{W\}$-normal form iff it is constructed from variables and blocks of the form $W_{\x_1}\circ \ldots W_{\x_n}\circ W_\z\circ\z\quad(n\geqslant 0)$ by application and abstraction. See Figure~\ref{sigmanf}.
\end{lemma}
\begin{proof}Induction over the structure of $A$. Suppose $A$ has the form $S\circ B$.\\
  If $B$ has the form $B_1B_2$, we can apply the rule $App$ and $A$ can not be a $\sigma$-normal form.\\
   If $B$ has the form $\lambda\x.B'$, we can apply the rule $Lambda$ and $A$ can not be a $\sigma$-normal form.\\
    Hence, by induction hypothesis, $B$ must be a variable or a block.\\[5pt]
Case 1. $B$ is a variable $\z$, hence $A$ is $S\circ\z$.\\
 If $S$ has the form $[C/\x]$, we can apply the rule $Var$ (if $\x=\z$) or the rule $Shift'$ (if $\x\neq\z$).\\
  If $S$ has the form $\{\y\x\}$, we can apply the rule $IdVar$ or the rule $IdShift'$.\\
   If $S$ has the form $\Up S'_\x$, we can apply the rule $LiftVar$ or the rule $LiftShift'$.\\
   If $S$ has the form $W_\x$ and $\x\neq\z$, we can apply the rule $W$.\\
    Hence, $S$ must be $W_\z$ and $A$ is $W_\z\circ\z$.\\[5pt]
Case 2. $B$ is a block and has the form $W_\x\circ B'$, hence $A$ is $S\circ W_\x\circ B'$\\
 (where $B'$ is a block or the variable $\x$).\\
 If $S$ has the form $[C/\x]$, we can apply the rule  $Shift$.\\
  If $S$ has the form $\{\y\x\}$, we can apply the rule  $IdShift$.\\
   If $S$ has the form $\Up S'_\x$, we can apply the rule  $LiftShift$.\\
    Hence, $S$ must has the form $W_\y$ and $A$ is the block $W_\y\circ B$.\\
\end{proof}

\begin{theorem}If $A$ is a good term and $A$ is a $\sigma\cup\{W,\alpha\}$-normal form, then $A$ is a usual lambda-term (i.e. without explicit substitutions).
\end{theorem}
\begin{proof}By Lemma~\ref{nfW}, it is sufficient to prove that $A$ does not contain blocks. Suppose $A$ contains a block $W_{\x_1}\circ \ldots W_{\x_n}\circ W_\z\circ\z$. By Generation lemma, the derivation of $G\vdash A$ contains a sub-derivation of  the form
 $$\ruledot{\qquad\qquad\Gamma,\z,\x_n\ldots\x_1\vdash W_{\x_1}\circ \ldots W_{\x_n}\circ W_\z\circ\z}$$
Below this judgement we  use only the rules $R4$ and $R5$ from Figure~\ref{inf1}.\\
 Note that\\
$FV(W_{\x_1}\circ \ldots W_{\x_n}\circ W_\z\circ\z)=\{\z\},\z,\x_n\ldots\x_1$\\
 Suppose $B_1$ is a well-formed term, constructed from $W_{\x_1}\circ \ldots W_{\x_n}\circ W_\z\circ\z$ and something else by application. Then $FV(B_1)$ has the form $\Delta_1,\z,\x_n\ldots\x_1$ and  $\z\in\Delta_1$. After the first application of $R5$ we obtain

$$\ruledot{\ruleone{\Gamma,\z,\x_n\ldots\x_1\vdash B_1}
{\Gamma,\z,\x_n\ldots\x_2\vdash\lambda\x_1.B_1}\qquad\qquad}$$
 and  $FV(\lambda\x_1.B_1)=\Delta_1,\z,\x_n\ldots\x_2$\\
Suppose $B_2$ is a well-formed term, constructed from $\lambda\x_1.B_1$ and something else by application. Then $FV(B_2)$ has the form $\Delta_2,\z,\x_n\ldots\x_2$ and $\z\in\Delta_2$. After the second application of $R5$ we obtain

$$\ruledot{\ruleone{\Gamma,\z,\x_n\ldots\x_2\vdash B_2}
{\Gamma,\z,\x_n\ldots\x_3\vdash\lambda\x_2.B_2}\qquad\qquad}\qquad$$
 and  $FV(\lambda\x_2.B_2)=\Delta_2,\z,\x_n\ldots\x_3$\\
After the $n+1$-th application of $R5$ we obtain

$$\ruledot{\quad\ruleone{\Gamma,\z\vdash B_{n+1}}
{\Gamma\vdash\lambda\z.B_{n+1}}}\qquad$$
where $\lambda\z.B_{n+1}$ is a subterm of $A$. $FV(B_{n+1})$ has the form $\Delta_{n+1},\z$   and\\ $\z\in\Delta_{n+1}$. But then $\z\in FV(\lambda\z.B_{n+1})$ and we can apply  the rule $\alpha$ to $\lambda\z.B_{n+1}$, hence $A$ can not be a $\sigma\cup\{W,\alpha\}$-normal form.\\
\end{proof}
\newpage

\section{$\sigma\cup\{W,\alpha\}$ is strongly normalizing}
\begin{figure}
\begin{framed}
\noindent
\textbf{Syntax.} $\Lambda\upsilon''$ is the set of terms inductively defined by the following BNF:

\begin{align*}
& &\x::&= x\mid y\mid z\mid\ldots \tag{Variables}\\
& &a,b::&= \x\mid \un \mid ab \mid \lambda a  \mid \ulam a\mid  a[s]  \tag{Terms}\\
& &s::&=   b/ \mid \,\,\uparrow\,\, \mid id \mid \,\,\,\Up s  \tag{Substitutions}
\end{align*}
\textbf{Rewrite rules.}

\begin{align*}
&App & (ab)[s] &\to (a[s])(b[s]) \\
&Lambda & (\lambda a)[s] &\to\lambda  (a[\Up s]) \\
&Lambda' & (\lambda a)[s] &\to\ulam (a[\Up s]) \\
&Lambda'' & (\ulam a)[s] &\to \lambda (a[\Up s]) \\
&Lambda''' & (\ulam a)[s] &\to\ulam (a[\Up s]) \\
&Var &  \un[b/] &\to b \\
&Shift & a[\uparrow][ b/] &\to a \\
&VarId &  \un[id] &\to \un \\
&ShiftId &  a[\uparrow][id] &\to a[\uparrow] \\
&VarLift & \un[\Up s]  &\to \un \\
&ShiftLift & a[\uparrow][\Up s] &\to a[s][\uparrow]\\
&\alpha & \ulam a &\to \lambda (a[id])\\
&\xi & \ulam a &\to \lambda a
\end{align*}

\end{framed}
\caption{The calculus $\upsilon''$}\label{upsilon''}
\end{figure}

\begin{definition}The calculus $\upsilon''$ is shown on Figure~\ref{upsilon''}. It has a new kind of terms $\ulam a$ and five new rewrite rules $Lambda'$, $Lambda''$, $Lambda'''$, $\alpha$, and $\xi$.
\end{definition}

\begin{definition}We associate with every derivable judgement $\Gamma\vdash A$  some  $\Lambda\upsilon''$-term $\|\Gamma\vdash A\|$ as it is shown in Figure~\ref{correspondence}, but with the following changes for abstraction:

\begin{align*}
&&&\ruleone{\|\Gamma,\x\vdash A\|=  a}{\|\Gamma\vdash\lambda \x.A\|=   \lambda a} & (x\not\in FV(\lambda\x.A)\\[15pt]
&&&\ruleone{\|\Gamma,\x\vdash A\|=  a}{\|\Gamma\vdash\lambda \x.A\|=   \ulam a} & (x\in FV(\lambda\x.A))
\end{align*}
\end{definition}

\begin{example}$ $\\[10pt]
$$\ruleone
{\ruletwo
{\|\{x\},x\vdash W_x\tri\{x\}\|= \,\,\uparrow}
{\|\{x\}\vdash x\|=  x}
{\|\{x\},x\vdash W_x\circ x\|=  x[\uparrow]}}
{\|\{x\}\vdash\lambda x.W_x\circ x\|= \ulam (x[\uparrow])}$$
\end{example}

\begin{theorem}\label{reduce}If the following conditions hold
\begin{itemize}
    \item $A_0\overset{\sigma\cup\{W,\alpha\}}{\longrightarrow} A_1\overset{\sigma\cup\{W,\alpha\}}{\longrightarrow}\ldots\overset{\sigma\cup\{W,\alpha\}}{\longrightarrow} A_n\overset{\sigma\cup\{W,\alpha\}}{\longrightarrow}\ldots$
  \item $\|\Gamma\vdash A_n\|=  a_n\quad\forall n\in\mathbb{N}$
\end{itemize}
then we get

\begin{itemize}
  \item [] $a_0\overset{\upsilon''}{\twoheadrightarrow}a_1\overset{\upsilon''}{\twoheadrightarrow}\ldots\overset{\upsilon''}{\twoheadrightarrow}a_n\overset{\upsilon''}{\twoheadrightarrow}\ldots$
\end{itemize}
 \end{theorem}

 \begin{proof}Recall that the rule $\alpha$ in $\sigma\cup\{W,\alpha\}$ is as follows\\
 $(\alpha)\quad \lambda\x. A \to \lambda \y.\{\y \x\} \circ A$\quad where\quad $\x\in FV(\lambda\x.A)\,\&\, \y\not\in FV(\lambda\x.A)$\\
 It corresponds to the rule $\alpha$ of $\upsilon''$.\\
By Theorem~\ref{decreaseFV}, $A\overset{\sigma\cup\{W,\alpha\}}{\longrightarrow} B$ implies $FV(A)\geqslant FV(B)$, hence if $\x\not\in FV(\lambda\x.A)$ then $\x\not\in FV(\lambda\x.B)$. Hence $\lambda a$ can not go to $\ulam b$ when we rewrite under lambda.
\end{proof}

\begin{example}$ $\\
$\{x\}\vdash\lambda x.W_x\circ x\overset{\alpha}{\to}\lambda y.\{yx\}\circ W_x\circ x\overset{IdShift}{\to}\lambda y.W_y\circ x\overset{W}{\to}\lambda y.x$\\
goes to\\
$0\vdash\ulam (x[\uparrow])\overset{\alpha}{\to}\lambda (x[\uparrow][id])\overset{IdShift}{\to}\lambda (x[\uparrow])$
\end{example}

To prove that $\upsilon''$ is strongly normalizing, we use the method of semantic
labelling. See~\cite{Zantema}.
\begin{definition}
To each term $a$ and each substitution $s$ we put in correspondence  natural numbers  (weights) $\|a\|$ and $\|s\|$  defined as follows:\\
\begin{align*}
 \|\x\|&=0\\
 \|\un\|&=0\\
 \|ab\|&=max(\|a\|,\|b\|)\\
\|\lambda a\|&=\|a\|+1\\
\|\ulam a\|&=\|a\|+1\\
\|a[s]\|&=\|a\|+\|s\|\\
\|b/\|&= \|b\|\\
\|\uparrow\|&=0\\
\|id\|&=0\\
\|\!\Up s\|&=\|s\|
 \end{align*}
\end{definition}
Note that all functional symbols of $\upsilon''$ (application, $\lambda$, $\ulam$, $-[-]$, $-/$, $\Up\,\,$)  turn to
 monotone functions of $\mathbb{N}$ to
$\mathbb{N}$ or of $\mathbb{N}\times\mathbb{N}$ to $\mathbb{N}$.

\begin{figure}
\begin{framed}
\noindent
\textbf{Syntax.} $\Lambda\upsilon'''$ is the set of terms inductively defined by the following BNF ($i\in\mathbb{N}$):

\begin{align*}
& &\x::&= x\mid y\mid z\mid\ldots \tag{Variables}\\
& &a,b::&= \x\mid \un \mid ab \mid \lambda a  \mid \ulam_i a\mid a\cdot_i[a]  \tag{Terms}\\
& &s::&=   b/ \mid \,\,\uparrow\,\, \mid id \mid \,\,\,\Up s  \tag{Substitutions}
\end{align*}
\textbf{Rewrite rules.}

\begin{align*}
&App & (ab)\cdot_{max(i,j)}[s] &\to (a\cdot_i[s])(b\cdot_j[s]) \\
&Lambda & (\lambda a)\cdot_{k+1}[s] &\to\lambda(a\cdot_k[\Up s]) \\
&Lambda' & (\lambda a)\cdot_{k+1}[s] &\to\ulam_{k+1}(a\cdot_k[\Up s]) \\
&Lambda'' & (\ulam_{i+1} a)\cdot_{i+j+1}[s] &\to\lambda(a\cdot_{i+j}[\Up s]) \\
&Lambda''' & (\ulam_{i+1} a)\cdot_{i+j+1}[s] &\to\ulam_{i+j+1}(a\cdot_{i+j}[\Up s]) \\
&Var & \un\cdot_i[b] &\to b \\
&Shift & a\cdot_i[\uparrow]\cdot_{i+j}[b/] &\to a \\
&VarId & \un\circ_0[id] &\to \un \\
&ShiftId &  a\cdot_i[\uparrow]\cdot_i[id] &\to a\cdot_i[\uparrow] \\
&VarLift & \un\cdot_i[\Up s] &\to \un \\
&ShiftLift & a\cdot_i[\uparrow]\cdot_{i+j}[\Up s] &\to a\cdot_{i+j}[s]\cdot_{i+j}[\uparrow]\\
&\alpha & \ulam_{i+1} a &\to \lambda(a\cdot_i[id])\\
&\xi & \ulam_i a &\to \lambda a\\
&Decr_1 & \ulam_i a &\to \ulam_j a \tag{$i>j$}\\
&Decr_2 & a\cdot_i[s] &\to a\cdot_j[s] \tag{$i>j$}
\end{align*}

\end{framed}
\caption{The calculus $\upsilon'''$}\label{upsilon'''}
\end{figure}

\begin{definition}The calculus $\upsilon'''$ is shown in Figure~\ref{upsilon'''}. It differs from $\upsilon''$ by the presence of natural indexes in $\ulam_i a$ and $a\cdot_i[s]$. The rules of $\upsilon'''$ are the rules of $\upsilon''$, where all terms $\ulam a$ and $a[s]$ are labelled by theirs weights (there are also  new rules $Decr_1$ and $Decr_2$).
\end{definition}

\begin{theorem}\label{normalizing}$\upsilon'''$ is strongly normalizing.
\end{theorem}
\begin{proof}
By choosing the  well-founded precedence\\
\begin{align*}
\cdot_i&>application\\
\cdot_i&>\cdot_j \qquad\,(i>j)\\
\cdot_i&>\lambda  \\
\cdot_i&>\,\, \Up\\
\cdot_i&>\ulam_i\\
\Up\,&>\,\,\uparrow\\
\ulam_i &>\lambda \\
\ulam_{i+1} &>\cdot_i\\
\ulam_i&> id \\
 \ulam_i&>\ulam_j \qquad
 (i>j)
\end{align*}
termination is easily proved by the lexicographic path order.\\
\end{proof}
\begin{theorem}\label{SN''}$\upsilon''$ is strongly normalizing.
\end{theorem}
\begin{proof}For any infinite sequence\\
\begin{itemize}
  \item []$a_0\overset{\upsilon''}{\to}a_1\overset{\upsilon''}{\to}a_2\overset{\upsilon''}{\to}\ldots\overset{\upsilon''}{\to}a_n\overset{\upsilon''}{\to}\ldots$
\end{itemize}
we can get an infinite sequence\\ \begin{itemize}
  \item []$a'_0\overset{\upsilon'''}{\twoheadrightarrow}a'_1\overset{\upsilon'''}{\twoheadrightarrow}a'_2\overset{\upsilon'''}{\twoheadrightarrow}\ldots\overset{\upsilon'''}{\twoheadrightarrow}a'_n\overset{\upsilon'''}{\twoheadrightarrow}\ldots$
\end{itemize}
by labelling all subterms of the forms $\ulam a$ and $a[s]$ by their weights. See~\cite{Zantema} (Theorem 81) for details.\\
\end{proof}
\begin{theorem}$\sigma\cup\{W,\alpha\}$ is strongly normalizing on the set of well-formed terms.
\end{theorem}
\begin{proof}By Theorem~\ref{reduce} and Theorem~\ref{SN''}.\\
\end{proof}

\end{document}